\documentclass[journal, twoside, web]{Config}

\usepackage{style, cite, textcomp, amsmath, amssymb, amsfonts, algorithm, algpseudocode, tikz}
\usepackage{graphicx, siunitx, bm, glossaries, float, booktabs, xcolor}
\usepackage[colorlinks=true, allcolors=black]{hyperref}
\usepackage[caption=false, font=footnotesize]{subfig}
\usepackage{lipsum}

\def\BibTeX{{\rm B\kern-.05em{\sc i\kern-.025em b}\kern-.08em
    T\kern-.1667em\lower.7ex\hbox{E}\kern-.125emX}}

\newtheorem{theorem}{Theorem}
\newtheorem{lemma}{Lemma}
\newtheorem{definition}{Definition}
\newtheorem{remark}{Remark}

\begin{document}

\title{Hybrid Lyapunov and Barrier Function-Based Control with Stabilization Guarantees}

\author{Hugo Matias, Daniel Silvestre
\thanks{This work was partially supported by the Portuguese Fundação para a Ciência e a Tecnologia 
(FCT) through the FirePuma project (DOI: 10.54499/PCIF/MPG/0156/2019), through the LARSyS FCT 
funding (DOI: 10.54499/LA/P/0083/2020, 10.54499/UIDP/50009/2020, and 10.54499/UIDB/50009/2020), 
and through the COPELABS, Lusófona University project UIDB/04111/2020.}
\thanks{H. Matias is with the Center of Technology and Systems (UNINOVA-CTS), NOVA School of 
Science and Technology (NOVA-FCT), 2829-516 Caparica, Portugal, and also with the Institute for 
Systems and Robotics (ISR-Lisbon), Instituto Superior Técnico (IST), 1049-001 Lisbon, Portugal 
(email: h.matias@campus.fct.unl.pt).}
\thanks{D. Silvestre is with the Center of Technology and Systems (UNINOVA-CTS), NOVA School of 
Science and Technology (NOVA-FCT), 2829-516 Caparica, Portugal, with the Institute for 
Systems and Robotics (ISR-Lisbon), Instituto Superior Técnico (IST), 1049-001 Lisbon, Portugal, and 
also with the COPELABS, Lusófona University, 1749-024 Lisbon, Portugal (e-mail: 
dsilvestre@isr.tecnico.ulisboa.pt).}}

\maketitle


\begin{abstract}
Control Lyapunov Functions (CLFs) and Control Barrier Functions (CBFs) can be combined, typically 
by means of Quadratic Programs (QPs), to design controllers that achieve performance and safety 
objectives. However, a significant limitation of the standard CLF-CBF-based framework is the 
emergence of unwanted equilibrium points that lead to deadlock situations and preclude task 
completion, even for simple systems and convex unsafe sets. To overcome this limitation, we propose 
a hybrid CLF-CBF control strategy that ensures safe global asymptotic stabilization, providing a 
more flexible and systematic design approach compared to current alternatives in the literature. 
The proposed method is first presented for first-order systems and then extended to higher-order 
systems via a joint CLF-CBF backstepping approach. The proposed solution is assessed through 
several simulation examples.
\end{abstract}



\section{Introduction}

\IEEEPARstart{S}{afety} is becoming an increasingly important consideration in modern control 
systems as these systems are being deployed in numerous real-world applications. Several control 
tasks require the design of controllers that achieve performance objectives, such as stabilization 
to a fixed point, while ensuring the system remains within a safe region throughout the control 
process. Main examples of such applications include obstacle avoidance for autonomous vehicles, 
automatic cruise control with lane keeping for automotive vehicles, and dynamic walking on uneven 
terrain for legged robots \cite{pandey2017mobile}, \cite{he2019adaptive}, \cite{reher2021dynamic}. 
However, maintaining the system within a specified region of the state space can also be useful for 
avoiding areas prone to significant disturbances, safe learning of system dynamics, and adaptive 
safety in the presence of parametric model uncertainty \cite{brunke2022safe}, 
\cite{hovakimyan20111}. 


\subsection{Literature Review}

Designing feedback laws that provide both asymptotic stabilization and safety guarantees is a 
complex task. Traditionally, this has been tackled using potential field methods, which have seen 
widespread use, particularly in the robotics field \cite{krogh1984generalized}, 
\cite{khatib1986real}. A potential field is a scalar function whose gradient drives the system 
toward the desired equilibrium point and away from the unsafe set. However, a significant drawback 
of potential field methods is the emergence of undesired equilibrium points that trap the system 
and prevent global convergence. To overcome this issue, Navigation Functions (NFs) have been 
proposed as an improved potential field method, achieving almost global asymptotic stabilization in 
sphere worlds by ensuring that the unwanted equilibrium points are unstable 
\cite{koditschek1990robot}, \cite{rimon1992exact}. Nonetheless, these methods are mostly limited to 
single-integrator dynamics and lack the flexibility to handle additional constraints. Also, they 
are prone to other drawbacks, such as oscillatory behavior, conservative trajectories, and 
difficult parameter tuning \cite{koren1991potential}.

Over the last decades, Model Predictive Control (MPC) has become a popular control technique, used 
in numerous applications. MPC optimizes a cost function over a prediction horizon at each discrete 
time step while directly incorporating state and input constraints \cite{schwenzer2021review}. 
However, nonlinear system dynamics and nonconvex sets for the admissible states render MPC into the 
class of nonconvex optimization with all its challenges. For instance, ellipsoidal representations 
of the unsafe sets result in nonconvex quadratic constraints in the MPC formulation 
\cite{silvestre2023model}. Conversely, if the unsafe sets are modeled as polytopes, the 
optimization problem becomes a mixed-integer program \cite{stoican2022mixed}. Consequently, despite 
its predictive advantages, MPC adds a significant computational load for real-time applications, 
even when resorting to convexification techniques \cite{mao2017successive}, \cite{gros2020linear}, 
\cite{taborda2024convex}.

In the past decade, Control Barrier Functions (CBFs) have emerged as a novel tool for designing 
controllers with formal safety guarantees for nonlinear systems \cite{ames2019control}. CBFs 
generalize Control Lyapunov Functions (CLFs) \cite{sontag1983lyapunov} for safety, where the key 
point is imposing a Lyapunov-like condition on the CBF time derivative to make the safe set forward 
invariant. Initially, CBFs were formulated as reciprocal barrier functions, which have the 
disadvantage of being unbounded at the boundary of the safe set, causing numerical problems 
\cite{ames2014control}. Later on, CBFs were introduced in the form of zeroing barrier functions. 
These become zero at the boundary of the safe set and provide better numerical properties, becoming 
standard in the field \cite{xu2015robustness}.

For control-affine systems, CLFs and CBFs can be unified by means of Quadratic 
Programs (QPs), effectively combining stabilization and safety objectives in a suitable control 
framework \cite{ames2016control}. Under this approach, the conditions for asymptotic stabilization 
and safety arise from the time derivatives of the CLF and CBF, which depend linearly on 
the control input. As a result, controllers can be synthesized using QPs with linear 
inequality constraints, which admit closed-form solutions \cite{li2023graphical}. The effectiveness 
of this method has been demonstrated across several applications \cite{hsu2015control}, 
\cite{wang2016safety}, \cite{wang2018safe}, \cite{taylor2020adaptive}. 

Since the original approach assumes that the control input directly affects the first-order time 
derivative of the CBF, some research has focused on its extension to higher-order systems. One 
direction involves High-Order Control Barrier Functions (HOCBFs), where a Lyapunov-like safety 
condition is imposed on a higher-order time derivative of the CBF that is affected by the input 
\cite{nguyen2016exponential}, \cite{xiao2021high}, \cite{tan2021high}. Alternatively, one can 
devise CBFs for higher-order systems via backstepping - a well-established technique for CLF design 
in cascaded systems \cite{sepulchre2012constructive}, which has been adapted for CBFs 
\cite{taylor2022safe}. This method simplifies the design process by recursively building a CBF for 
the full system from a CBF designed only for the top-level subsystem.

Research has also focused on extending this framework to handle more complex safety objectives. 
Some works directly incorporate multiple CBFs into the control design by enforcing multiple CBF 
constraints within the QP formulation \cite{rauscher2016constrained}, 
\cite{breeden2023compositions}. Other works merge complex safety requirements into a single CBF, 
usually via Boolean logic operations, such as AND, OR, and negation, initially established using 
nonsmooth CBFs \cite{glotfelter2017nonsmooth}, \cite{ glotfelter2020nonsmooth}. Most recently, in 
\cite{molnar2023composing}, the authors propose an algorithmic scheme to design a single smooth CBF 
through Boolean logic and smooth maximum and minimum approximations, which can address multiple 
logical compositions of safety constraints.

However, the standard CLF-CBF-QP-based framework has a major drawback. While it enforces the 
forward invariance of the safe set with a hard constraint, it relaxes the stabilization objective 
to maintain the feasibility of the QP. This leads to the emergence of additional equilibrium 
points beyond the zero of the CLF. In particular, undesired asymptotically stable points can 
appear at the boundary of the safe set, leading to deadlocks and hindering task completion 
\cite{reis2020control}, \cite{grover2020does}. This closely mirrors the similar problem found in 
potential field methods.

Some works have addressed deadlock resolution in the CBF-based context. One of them introduces the 
concept of Control Lyapunov Barrier Functions (CLBFs) to guarantee safety and global asymptotic 
stabilization \cite{romdlony2016stabilization}. However, as demonstrated in 
\cite{braun2020comment}, there cannot exist a CLBF that renders a point globally asymptotically 
stable while avoiding a bounded set. Also, in \cite{cortez2022compatibility}, the authors identify 
a set where convergence to the origin is guaranteed. However, as the region of attraction of a 
point within a continuous vector field must be diffeomorphic to the Euclidean space, the boundary 
of a bounded unsafe set cannot be fully included in the region of attraction 
\cite{poveda2021robust}.

In \cite{reis2020control}, the authors show that the standard CLF-CBF-QP-based approach can 
introduce undesired asymptotically stable equilibrium points and propose eliminating them by 
combining CBFs with radially-asymmetric rotating CLFs. Nevertheless, despite making 
the unwanted equilibrium points unstable, the authors do not provide global convergence guarantees. 
Also, this approach may lead to undesired oscillations. Later, in \cite{tan2024undesired}, 
the authors show that the unwanted interior equilibrium points can be eliminated using a nominal 
stabilizing controller but do not address the more challenging boundary equilibria.

Recently, deadlock resolution in the context of CBF-based control has been addressed via hybrid 
feedback, as continuous control approaches present a greater difficulty in handling this issue 
\cite{sanfelice2021hybrid}. In particular, hybrid CBF formulations, originally introduced for 
hybrid systems \cite{glotfelter2019hybrid}, \cite{robey2021learning}, have been proposed for 
deadlock resolution in continuous-time systems, where decisiveness is achieved by enhancing CBFs 
with logic variables. In \cite{braun2020explicit}, \cite{braun2021augmented}, and 
\cite{ballaben2024lyapunov}, the authors propose using an avoidance shell described by two 
partially overlapping domains, where decisiveness is achieved by switching between these domains. 
However, despite ensuring global asymptotic stabilization and safety, the avoidance shell may 
provide an overly conservative representation of the actual unsafe set. 

In \cite{marley2024hybrid}, the authors propose a hybrid CBF approach based on a polytopic 
avoidance domain, which may provide a less conservative description of the unsafe set. Under this 
method, the trajectory sequentially converges to the induced equilibria on each active hyperplane, 
resolving deadlocks via a switching mechanism similar to that from synergistic Lyapunov functions 
\cite{mayhew2011synergistic}, \cite{mayhew2011further}. However, the design from 
\cite{marley2024hybrid} also lacks flexibility, as deadlock resolution is confined to specific 
polytopes where all the induced equilibrium points are contained by more than one safe half-space. 
Thus, a particular polytope would have to be designed, which bounds the actual unsafe set and 
verifies this condition. Moreover, this polytopic description would only be valid for a specific 
set of desired equilibrium points since the positions of the induced equilibria depend on the 
desired equilibrium point through the CLF. 

Most recently, the work in \cite{mestres2024safe} introduces a sampling-based method that combines 
CLFs and CBFs with Rapidly-Exploring Random Trees (RRTs) to generate CLF-CBF compatible paths that 
guide the system toward the desired equilibrium point. However, while effective, this method has a 
higher computational complexity and has a probabilistic completeness nature. In contrast, our goal 
is to pursue a more elementary, geometric hybrid approach, following the lines of research of the 
hybrid approaches from \cite{braun2021augmented} and \cite{marley2024hybrid}.


\subsection{Paper Overview}

Inspired by the hybrid approaches from \cite{braun2021augmented} and \cite{marley2024hybrid}, this 
paper introduces a hybrid CLF-CBF control framework with global asymptotic stabilization and safety 
guarantees, offering a more flexible and systematic design approach. The proposed solution relies 
on a polytopic avoidance domain, which can be any bounded convex polytope that encloses the unsafe 
set, and it involves solving a series of safe stabilization subproblems. Each subproblem involves 
an active safe half-space and active target point, for which we show how to build compatible 
CLF and CBF constraints, ensuring convergence to the active target point. As the system 
approaches the active target, a switching mechanism updates the active half-space and active 
target, and global asymptotic stabilization is achieved by ensuring that the switching logic 
generates a setpoint sequence that converges to the desired equilibrium point. The approach is 
presented for first-order systems and extended to higher-order systems via a joint CLF-CBF 
backstepping procedure.

The remainder of this paper is structured as follows. Section \ref{Sec:Background} provides 
essential preliminaries, illustrative examples, and the problem statement. Sections 
\ref{Sec:HybridController} and \ref{Sec:Backstepping} present the hybrid control solution, with 
simulation results in Section \ref{Sec:Results}. Finally, Section \ref{Sec:Conclusion} summarizes 
conclusions and future directions.


\subsection{Notation}

$\mathbb{N}$ is the nonnegative integers set. $\mathbb{R}$, $\mathbb{R}_{\geq 0}$, and 
$\mathbb{R}_{>0}$ are the sets of real, nonnegative, and positive numbers, respectively. 
$\mathbb{R}^n$ is the $n$-dimensional euclidean space, and $\mathbb{S}^{n-1}$ is the unit sphere in 
$\mathbb{R}^n$. $\mathbb{R}^{n\times m}$ is the set of $n\times m$ real matrices, 
$\mathbb{R}^{n\times n}_{\succ 0}$ is the set of positive-definite matrices of size $n$, and 
$\mathrm{SO}(n)$ is the special orthogonal group in $\mathbb{R}^n$. For some set 
$\mathcal{S} \subseteq \mathbb{R}^n$, $\mathrm{int}(\mathcal{S})$ and $\partial\mathcal{S}$ are the 
interior and boundary of $\mathcal{S}$, respectively. The $p$-norm of a vector 
$\mathbf{x} \in \mathbb{R}^n$ is denoted $\|\mathbf{x}\|_p$ ($\|\mathbf{x}\| = \|\mathbf{x}\|_2$), 
and for two vectors $\mathbf{x}_1 \in\mathbb{R}^{n_1}$, $\mathbf{x}_2 \in \mathbb{R}^{n_2}$, we 
often use the notation 
$(\mathbf{x}_1, \mathbf{x}_2) = \left[\mathbf{x}_1^\top\, \mathbf{x}_2^\top\right]^\top \in 
\mathbb{R}^{n_1+n_2}$. For a differentiable function $h: \mathbb{R}^n \rightarrow \mathbb{R}$ and 
a field $\mathbf{G}: \mathbb{R}^n \rightarrow \mathbb{R}^{n \times m}$, we consider the 
Lie-derivative notation 
$L_\mathbf{G}h(\mathbf{x}) = \nabla h(\mathbf{x})^\top \mathbf{G}(\mathbf{x})$, and $\mathbf{x}^+$ 
is the value of $\mathbf{x}$ after an instantaneous change. A scalar function 
$V: \mathbb{R}^n \rightarrow \mathbb{R}_{\geq 0}$ is said to be positive definite around a point 
$\bar{\mathbf{x}}$ if $V(\bar{\mathbf{x}}) = 0$ and $V(\mathbf{x}) > 0$ for all 
$\mathbf{x} \in \mathbb{R}^n\setminus\{\bar{\mathbf{x}}\}$. Moreover, $\mathbf{0}_{n\times m}$ 
is the $n\times m$ zero matrix, and $\mathbf{I}_{n}$ is the identity matrix of size $n$ (dimensions 
may be omitted when clear from context). Finally, the $\arg\max$ and $\arg\min$ operators 
return an argument that optimizes some cost function over a feasible set.


\section{Preliminaries and Problem Statement} \label{Sec:Background}

We consider nonlinear control-affine systems of the form
\begin{equation}
    \Dot{\mathbf{x}} = \mathbf{f}(\mathbf{x}) + \mathbf{G}(\mathbf{x})\mathbf{u},
    \label{Eq:ControlAffineSystem}
\end{equation}
where $\mathbf{x} \in \mathbb{R}^n$ is the system state, $\mathbf{u} \in \mathbb{R}^m$ is the 
control input, and the fields $\mathbf{f}: \mathbb{R}^n \rightarrow \mathbb{R}^n$ and 
$\mathbf{G}: \mathbb{R}^n \rightarrow \mathbb{R}^{n \times m}$ are locally Lipschitz continuous. 
Applying a locally Lipschitz continuous controller 
$\mathbf{k}: \mathbb{R}^n \rightarrow \mathbb{R}^m$ to \eqref{Eq:ControlAffineSystem} yields the 
closed-loop system
\begin{equation}
    \Dot{\mathbf{x}} = \mathbf{f}(\mathbf{x}) + \mathbf{G}(\mathbf{x})\mathbf{k}(\mathbf{x}).
    \label{Eq:ControlAffineSystemClosedLoop}
\end{equation}
As the functions $\mathbf{f}$, $\mathbf{G}$, and $\mathbf{k}$ are locally Lipschitz continuous, for 
every initial condition $\mathbf{x}_0 \in \mathbb{R}^n$, there exists a unique continuously 
differentiable solution $\bm{\varphi}: I(\mathbf{x}_0) \rightarrow \mathbb{R}^n$ satisfying 
\begin{equation}
    \begin{aligned}
        \Dot{\bm{\varphi}}(t) &= \mathbf{f}(\bm{\varphi}(t)) 
        + \mathbf{G}(\bm{\varphi}(t))\mathbf{k}(\bm{\varphi}(t)),\\
        \bm{\varphi}(0) &= \mathbf{x}_0,
    \end{aligned}
    \label{Eq:ControlAffineSystemSolution}
\end{equation}
for all $t \in I(\mathbf{x}_0)$, where $I(\mathbf{x}_0) \subseteq \mathbb{R}_{\geq 0}$ is the 
maximal interval of existence for the solution \cite{perko2013differential}. If 
$I(\mathbf{x}_0) = \mathbb{R}_{\geq 0}$, the solution is called complete. Next, we define the 
notions of asymptotic stability and forward invariance considered in this paper.

\begin{definition}[Asymptotic Stability]
    An equilibrium point $\bar{\mathbf{x}}$ of the system \eqref{Eq:ControlAffineSystemClosedLoop} 
    is asymptotically stable if it is stable and there exists a maximal set 
    $\mathcal{A} \supset \{\bar{\mathbf{x}}\}$ such that, for every initial condition 
    $\mathbf{x}_0 \in \mathcal{A}$, $\bm{\varphi}$ is complete and 
    $\lim_{t\rightarrow\infty} \|\bm{\varphi}(t) - \bar{\mathbf{x}}\| = 0$. The set 
    $\mathcal{A}$ is called the region of attraction of $\bar{\mathbf{x}}$. If 
    $\mathcal{A} = \mathbb{R}^n$, then $\bar{\mathbf{x}}$ is said to be globally asymptotically 
    stable.
\end{definition}

\begin{definition}[Forward Invariance]
    A set $\mathcal{C} \subset \mathbb{R}^n$ is said to be forward invariant for the system 
    \eqref{Eq:ControlAffineSystemClosedLoop} if, for every initial condition 
    $\mathbf{x}_0 \in \mathcal{C}$, we have $\bm{\varphi}(t) \in \mathcal{C}$ for all 
    $t \in I(\mathbf{x}_0)$.
\end{definition}


\subsection{Control Lyapunov and Control Barrier Functions}

We begin by considering the common objective of globally asymptotically stabilizing the system 
\eqref{Eq:ControlAffineSystem} to a desired equilibrium point $\bar{\mathbf{x}}$. This can be 
achieved by designing a control law that drives a proper and positive-definite function (around 
$\bar{\mathbf{x}}$) $V: \mathbb{R}^n \rightarrow \mathbb{R}_{\geq 0}$ to zero, motivating the 
concept of CLF \cite{sontag1989universal}.

\begin{definition}[CLF]
    A continuously differentiable, proper, and positive-definite function 
    $V: \mathbb{R}^n \rightarrow \mathbb{R}_{\geq 0}$ around a point $\bar{\mathbf{x}}$ is a 
    CLF for the system \eqref{Eq:ControlAffineSystem} if there exists a class-$\mathcal{K}$ 
    function $\gamma: \mathbb{R}_{\geq 0} \rightarrow \mathbb{R}_{\geq 0}$ such that, for all 
    $\mathbf{x} \in \mathbb{R}^n\setminus\{\bar{\mathbf{x}}\}$,
    \begin{equation}
        \inf_{\mathbf{u} \in \mathbb{R}^m}
        [L_\mathbf{f}V(\mathbf{x}) + L_\mathbf{G}V(\mathbf{x})\mathbf{u}] < - \gamma(V(\mathbf{x})).
        \label{Eq:DefinitionCLF}
    \end{equation}
    
    Given a CLF $V$ for \eqref{Eq:ControlAffineSystem} and a corresponding 
    class-$\mathcal{K}$ function $\gamma$, we define the pointwise set of control vectors
    \begin{equation}
        \scalebox{0.915}{$
        K_\mathrm{CLF}(\mathbf{x}) = \{\mathbf{u} \in \mathbb{R}^m: 
        L_\mathbf{f}V(\mathbf{x}) + L_\mathbf{G}V(\mathbf{x})\mathbf{u} 
        \leq - \gamma(V(\mathbf{x}))\}$}.
        \label{Eq:SetCLF}
    \end{equation}
    This yields the following main result with respect to CLFs.
\end{definition}

\begin{theorem}[Stabilizing Control \cite{sontag1989universal}]
    Let $V: \mathbb{R}^n \rightarrow \mathbb{R}_{\geq 0}$ be a continuously differentiable, proper, 
    and positive-definite function around a point $\bar{\mathbf{x}}$. If $V$ is a CLF for 
    \eqref{Eq:ControlAffineSystem}, then the set $K_\mathrm{CLF}(\mathbf{x})$ is nonempty for all 
    $\mathbf{x} \in \mathbb{R}^n$, and a locally Lipschitz controller 
    $\mathbf{k}: \mathbb{R}^n \rightarrow \mathbb{R}^m$ with 
    $\mathbf{k}(\mathbf{x}) \in K_\mathrm{CLF}(\mathbf{x})$ for all $\mathbf{x} \in \mathbb{R}^n$ 
    globally asymptotically stabilizes the system to $\bar{\mathbf{x}}$.
\end{theorem}

In addition, we consider the goal of making a given safe set forward invariant. Particularly, we 
consider a set $\mathcal{C} \subset \mathbb{R}^n$ that is the 0-superlevel set of a continuously 
differentiable function $h: \mathbb{R}^n \rightarrow\mathbb{R}$ with 
$\nabla h(\mathbf{x}) \neq\mathbf{0}$ when $h(\mathbf{x}) = 0$, yielding
\begin{equation}
    \begin{aligned}
        \mathcal{C} &= \{\mathbf{x} \in \mathbb{R}^n: h(\mathbf{x}) \geq 0\},\\
        \partial\mathcal{C} &= \{\mathbf{x} \in \mathbb{R}^n: h(\mathbf{x}) = 0\},\\
        \mathrm{int}(\mathcal{C}) &= \{\mathbf{x} \in \mathbb{R}^n: h(\mathbf{x}) > 0\}.
    \end{aligned}
    \label{Eq:SafeSet}
\end{equation}
If $h$ has the properties of a CBF, then it can be used to design safe controllers for the system 
\eqref{Eq:ControlAffineSystem}.

\begin{definition}[CBF \cite{ames2016control}]
    Let $\mathcal{C} \subset \mathbb{R}^n$ be the 0-superlevel set of a continuously differentiable 
    function $h: \mathbb{R}^n \rightarrow\mathbb{R}$ with $\nabla h(\mathbf{x}) \neq\mathbf{0}$ 
    when $h(\mathbf{x}) = 0$. The function $h$ is a (zeroing) CBF for system 
    \eqref{Eq:ControlAffineSystem} if there exists an extended class-$\mathcal{K}_\infty$ function 
    $\alpha: \mathbb{R} \rightarrow \mathbb{R}$ such that, for all $\mathbf{x} \in \mathbb{R}^n$,
    \begin{equation}
        \sup_{\mathbf{u} \in \mathbb{R}^m}
        [L_\mathbf{f}h(\mathbf{x}) + L_\mathbf{G}h(\mathbf{x})\mathbf{u}] > - \alpha(h(\mathbf{x})).
        \label{Eq:DefinitionCBF}
    \end{equation}

    Such a definition means that a CBF is allowed to decrease in the interior of the safe set but 
    not on its boundary. Similar to CLFs, given a CBF $h$ for \eqref{Eq:ControlAffineSystem} and a 
    corresponding extended class-$\mathcal{K}_\infty$ function $\alpha$, we define the pointwise 
    set of controls
    \begin{equation}
        \scalebox{0.935}{$
        K_\mathrm{CBF}(\mathbf{x}) = \{\mathbf{u} \in \mathbb{R}^m: 
        L_\mathbf{f}h(\mathbf{x}) + L_\mathbf{G}h(\mathbf{x})\mathbf{u} 
        \geq - \alpha(h(\mathbf{x}))\}$}.
        \label{Eq:SetCBF}
    \end{equation}
    This yields the following main result concerning CBFs.
\end{definition}

\begin{theorem}[Safeguarding Controller \cite{ames2016control}]
    Let $\mathcal{C} \subset \mathbb{R}^n$ be the 0-superlevel set of a continuously 
    differentiable function $h: \mathbb{R}^n \rightarrow\mathbb{R}$ with 
    $\nabla h(\mathbf{x}) \neq\mathbf{0}$ when $h(\mathbf{x}) = 0$. If the function $h$ is a CBF 
    for system \eqref{Eq:ControlAffineSystem}, then the set $K_\mathrm{CBF}(\mathbf{x})$ is 
    nonempty for all $\mathbf{x} \in \mathbb{R}^n$, and any locally Lipschitz continuous controller 
    $\mathbf{k}: \mathbb{R}^n \rightarrow \mathbb{R}^m$ with 
    $\mathbf{k}(\mathbf{x}) \in K_\mathrm{CBF}(\mathbf{x})$ for all $\mathbf{x} \in \mathbb{R}^n$ 
    renders $\mathcal{C}$ forward invariant for the resulting closed-loop system. Also, the set
    $\mathcal{C}$ becomes asymptotically stable in $\mathbb{R}^n$.
\end{theorem}

\begin{remark}
    The strictness of the inequalities in \eqref{Eq:DefinitionCLF} and \eqref{Eq:DefinitionCBF} 
    enables proving that optimization-based controllers relying on CLFs and CBFs are 
    locally Lipschitz continuous \cite{morris2013sufficient}, \cite{jankovic2018robust}.
\end{remark}


\subsection{Quadratic Program Formulation}

Stabilization and safety objectives represented by CLFs and CBFs can be unified through an 
optimization-based approach based on QPs. More specifically, given a CLF $V$ and a CBF $h$ 
associated with a safe set, these objectives can be incorporated into a controller 
$\mathbf{k}: \mathbb{R}^n \rightarrow \mathbb{R}^m$ through the following QP:
\begin{equation} 
    \begin{aligned}
        (\mathbf{k}(\mathbf{x}), \cdot)
        =\,\, &\underset{(\mathbf{u}, \delta) \in \mathbb{R}^{m+1}}{\arg\min}\,\,
        \frac{1}{2}\|\mathbf{u}\|^2 + \frac{1}{2}p\delta^2\\
        \text{subject to}\,\, &L_{\mathbf{f}}V(\mathbf{x})+L_{\mathbf{G}}V(\mathbf{x})\mathbf{u} 
        \leq -\gamma(V(\mathbf{x})) + \delta,\\
        &L_{\mathbf{f}}h(\mathbf{x})+L_{\mathbf{G}}h(\mathbf{x})\mathbf{u} 
        \geq -\alpha(h(\mathbf{x})),
    \end{aligned}
    \label{Eq:CLF-CBF-QP}
\end{equation}
where $p \in \mathbb{R}_{>0}$ is a scaling parameter, $\gamma$ is a class-$\mathcal{K}$ function 
associated with the CLF, and $\alpha$ an extended class-$\mathcal{K}_\infty$ function associated 
with the CBF. The CBF constraint ensures forward invariance of the safe set, and the relaxation 
variable $\delta$ softens the stability objective to keep the feasibility of the optimization 
problem across all $\mathbf{x} \in \mathbb{R}^n$.

For compactness, we now let
$F_V(\mathbf{x}) = L_{\mathbf{f}}V(\mathbf{x}) + \gamma(V(\mathbf{x}))$, 
$F_h(\mathbf{x}) = L_{\mathbf{f}}h(\mathbf{x}) + \alpha(h(\mathbf{x}))$, 
and $L(\mathbf{x}) = L_{\mathbf{G}}V(\mathbf{x})L_{\mathbf{G}}h(\mathbf{x})^\top$. According to 
the Karush–Kuhn–Tucker (KKT) conditions, the QP controller can be expressed in closed-form as
\begin{equation}
    \mathbf{k}(\mathbf{x}) =
    \begin{cases}
        \mathbf{k}_1(\mathbf{x}), &\text{if }\mathbf{x} \in \mathcal{S}_1,\\
        \mathbf{k}_2(\mathbf{x}), &\text{if }\mathbf{x} \in \mathcal{S}_2,\\
        \mathbf{k}_3(\mathbf{x}), &\text{if }\mathbf{x} \in \mathcal{S}_3,\\
        \mathbf{0},               &\text{if }\mathbf{x} \in \mathcal{S}_4,
    \end{cases}
    \label{Eq:ControllerExpressionStart}
\end{equation}
where the expressions corresponding to each case are
\begin{align}
    \mathbf{k}_1(\mathbf{x}) &= -\left(p^{-1} + \|L_{\mathbf{G}}V(\mathbf{x})\|^2\right)^{-1}
    F_V(\mathbf{x})L_{\mathbf{G}}V(\mathbf{x})^\top,\nonumber\\
    \mathbf{k}_2(\mathbf{x}) &= -\|L_{\mathbf{G}}h(\mathbf{x})\|^{-2}
    F_h(\mathbf{x})L_{\mathbf{G}}h(\mathbf{x})^\top,\nonumber\\
    \mathbf{k}_3(\mathbf{x}) &= 
    -\lambda_1(\mathbf{x})L_{\mathbf{G}}V(\mathbf{x})^\top 
    +\lambda_2(\mathbf{x})L_{\mathbf{G}}h(\mathbf{x})^\top.
\end{align}
Also, concerning the third case, $\lambda_1(\mathbf{x})$ and $\lambda_2(\mathbf{x})$ are given by
\begin{equation}
    \scalebox{0.86}{$
    \begin{aligned}
        \lambda_1(\mathbf{x}) &= \Delta(\mathbf{x})^{-1}\left(L(\mathbf{x})F_h(\mathbf{x})
        - \|L_{\mathbf{G}}h(\mathbf{x})\|^2F_V(\mathbf{x})\right),\\
        \lambda_2(\mathbf{x}) &= \Delta(\mathbf{x})^{-1}\left(\left(
        p^{-1} + \|L_{\mathbf{G}}V(\mathbf{x})\|^2\right)F_h(\mathbf{x}) 
        - L(\mathbf{x})F_V(\mathbf{x})\right),
    \end{aligned}$}
\end{equation}
where $\Delta(\mathbf{x})$ is defined as
\begin{equation}
    \Delta(\mathbf{x}) = L(\mathbf{x})^2 
    - \left(p^{-1} + \|L_{\mathbf{G}}V(\mathbf{x})\|^2\right)\|L_{\mathbf{G}}h(\mathbf{x})\|^2.
\end{equation} 
Moreover, the subdomains defining each case are given by
\begin{equation}
    \begin{aligned}
        \mathcal{S}_1 &= \{\mathbf{x} \in \mathbb{R}^n: 
        F_V(\mathbf{x}) \geq 0,\, s_1(\mathbf{x}) > 0\},\\
        \mathcal{S}_2 &= \{\mathbf{x} \in \mathbb{R}^n: 
        F_h(\mathbf{x}) \leq 0,\, s_2(\mathbf{x}) < 0\},\\
        \mathcal{S}_3 &= \{\mathbf{x} \in \mathbb{R}^n: 
        \Delta(\mathbf{x}) \neq 0,\, 
        \lambda_1(\mathbf{x}) \geq 0,\, \lambda_2(\mathbf{x}) \geq 0\},\\
        \mathcal{S}_4 &= \{\mathbf{x} \in\mathbb{R}^n: 
        F_V(\mathbf{x}) < 0,\, F_h(\mathbf{x}) > 0\},
    \end{aligned}
\end{equation}
where $s_1(\mathbf{x})$ and $s_2(\mathbf{x})$ are defined as
\begin{equation}
    \begin{aligned}
        s_1(\mathbf{x}) &= \left(p^{-1} + \|L_{\mathbf{G}}V(\mathbf{x})\|^2\right)F_h(\mathbf{x}) 
        - L(\mathbf{x})F_V(\mathbf{x}),\\
        s_2(\mathbf{x}) &= \|L_{\mathbf{G}}h(\mathbf{x})\|^2F_V(\mathbf{x}) 
        - L(\mathbf{x}) F_h(\mathbf{x}).
    \end{aligned}
    \label{Eq:ControllerExpressionEnd}
\end{equation}

The first case corresponds to the CLF constraint being active and the CBF constraint being 
inactive. Conversely, the second scenario occurs when the CBF constraint is active, but the CLF 
constraint is inactive. The third case involves both constraints being active, and finally, the 
fourth case matches the scenario where neither constraint is active \cite{ames2016control}, 
\cite{li2023graphical}. Furthermore, the controller defined in \eqref{Eq:CLF-CBF-QP} is locally 
Lipschitz continuous if the  gradients of the CLF and CBF, along with $\gamma$ and $\alpha$, are 
locally Lipschitz continuous \cite{morris2013sufficient}, \cite{jankovic2018robust}.

\newpage

Nevertheless, this general approach has a significant drawback. Despite ensuring the forward 
invariance of the safe set as a strict requirement, relaxing the stabilization objective can 
introduce additional equilibrium points besides the minimizer of the CLF. Particularly, the 
set of equilibrium points on the safe set $\mathcal{C}$ of the closed-loop system that results 
from applying the controller \eqref{Eq:CLF-CBF-QP} into \eqref{Eq:ControlAffineSystem}, 
$\mathcal{E}_\mathcal{C}$, is determined as
\begin{equation}
    \mathcal{E}_\mathcal{C} = 
    \mathcal{E}_{\mathrm{int}(\mathcal{C})} \cup \mathcal{E}_{\partial\mathcal{C}},
\end{equation}
where $\mathcal{E}_{\mathrm{int}(\mathcal{C})}$ and $\mathcal{E}_{\partial\mathcal{C}}$ denote the 
sets of interior and boundary equilibria, respectively, which are given by
\begin{equation}
    \begin{aligned}
        \mathcal{E}_{\mathrm{int}(\mathcal{C})} &= 
        \{\mathbf{x} \in \mathcal{S}_1 \cap \mathrm{int}(\mathcal{C}): 
        \mathbf{f}(\mathbf{x}) + \mathbf{G}(\mathbf{x})\mathbf{k}_1(\mathbf{x}) = \mathbf{0}\},\\
        \mathcal{E}_{\partial\mathcal{C}} &= 
        \{\mathbf{x} \in \mathcal{S}_3 \cap \partial\mathcal{C}: 
        \mathbf{f}(\mathbf{x}) + \mathbf{G}(\mathbf{x})\mathbf{k}_3(\mathbf{x}) = \mathbf{0}\}.
    \end{aligned}
\end{equation}
Moreover, as detailed in \cite{reis2020control}, some of the induced equilibrium points can even 
be asymptotically stable, leading to deadlock situations and undermining any guarantees of task 
completion.


\subsection{Illustrative Examples}

In this subsection, we present a few examples that illustrate the application of the CLF-CBF-QP 
approach to an avoidance control problem. To provide a richer analysis and highlight the benefits 
and limitations of this method, we also include some examples obtained with an MPC approach. For 
simplicity, we consider a single-integrator system, described by
\begin{equation}
    \Dot{\mathbf{x}} = \mathbf{u},
    \label{Eq:SingleIntegrator}
\end{equation}
with $\mathbf{x}, \mathbf{u} \in \mathbb{R}^n$, and the objective is to asymptotically stabilize 
the system to a point $\bar{\mathbf{x}}$ while avoiding a bounded set 
$\mathcal{O} \subset \mathbb{R}^n$. To achieve this, we explore the strategies outlined below.

\subsubsection{CLF-CBF-QP - Ellipsoidal Fit}

This strategy represents the simplest and most straightforward approach to the problem. It consists 
of the CLF-CBF-QP formulation with the usual choice of a standard quadratic CLF $V$, so that
\begin{equation}
    V(\mathbf{x}) = \frac{1}{2}\|\mathbf{x} - \bar{\mathbf{x}}\|^2
    \label{Eq:IllustrativeExamplesCLF}
\end{equation}
for all $\mathbf{x} \in \mathbb{R}^n$, and it is based on an ellipsoidal approximation of the 
unsafe set. Accordingly, we define a safe set 
$\mathcal{C} \subseteq \mathbb{R}^n\setminus\mathcal{O}$ as the 0-superlevel set of a quadratic 
CBF $h$ given by
\begin{equation}
    h(\mathbf{x}) = \frac{1}{2}(\mathbf{x} - \mathbf{c})^\top\mathbf{A}(\mathbf{x} - \mathbf{c})
    - \frac{1}{2}r^2
\end{equation}
for all $\mathbf{x} \in \mathbb{R}^n$, where $\mathbf{c} \in \mathbb{R}^n$, 
$\mathbf{A} \in \mathbb{R}^{n\times n}_{\succ 0}$, and $r \in \mathbb{R}_{>0}$.

\subsubsection{CLF-CBF-QP - Polytopic Fit} \label{Sec:IllustrativeExamplesSmoothMax}

Alternatively, another approach that can be considered is approximating the unsafe set with a 
convex polytope, which may provide a less conservative representation of more complex unsafe 
regions. However, as the complement of a convex polytope is a union of multiple half-spaces, it 
cannot be directly defined as the 0-superlevel set of a single CBF. To overcome this, we adopt the 
technique presented in \cite{molnar2023composing} and establish a single CBF through a smooth 
approximation of the maximum function. More precisely, following the approach from 
\cite{molnar2023composing}, we define a safe set 
$\mathcal{C} \subseteq \mathbb{R}^n\setminus\mathcal{O}$ as the 0-superlevel set of a CBF $h$ given 
by
\begin{equation}
    h(\mathbf{x}) = \frac{1}{\kappa}\ln\left(\frac{1}{Q}\sum_{q=1}^{Q}\exp\left(\kappa 
    \left(\mathbf{n}_q^\top\mathbf{x}-d_q\right)\right)\right)
    \label{Eq:SmoothMax}
\end{equation}
for all $\mathbf{x} \in \mathbb{R}^n$, where $\kappa \in \mathbb{R}_{> 0}$ is a smoothing parameter 
and,

\newpage

\noindent for each $q \in \{1, \dots, Q\}$, $\mathbf{n}_q \in \mathbb{S}^{n-1}$ and 
$d_q \in \mathbb{R}$ denote, respectively, the unit outward normal and the offset associated with 
one of the facets of a convex polytope that is enclosed by the complement of $\mathcal{C}$. 
Additionally, we also consider a standard quadratic CLF $V$, defined by 
\eqref{Eq:IllustrativeExamplesCLF} for all $\mathbf{x} \in \mathbb{R}^n$.

\subsubsection{MPC - Polytopic Fit}

Finally, for comparison, we consider an MPC approach to the problem, where the unsafe set is also 
modeled as a convex polytope. Particularly, we consider a mixed-integer formulation in which, at 
each discrete-time instant, safety is enforced by requiring that at least one half-space constraint 
is satisfied (see e.g. \cite{stoican2022mixed}).

Fig. \ref{Fig:IllustrativeExamples} shows a few examples of system trajectories obtained by 
applying the previously mentioned strategies to avoid two different unsafe sets in a 
two-dimensional setting ($n = 2$). For simplicity and to facilitate the application of each method, 
we have considered polytopic unsafe sets $\mathcal{O} \subset \mathbb{R}^2$.

\newcommand{\xchaserinit}{
    \begin{tikzpicture}
        \protect\draw[black, thick] (-0.05,-0.05) -- (0.05,0.05);
        \protect\draw[black, thick] (-0.05,0.05) -- (0.05,-0.05);
    \end{tikzpicture}
}
\newcommand{\xtarget}{
    \begin{tikzpicture}
        \protect\draw[black, thick, fill] (0,0) circle (1.5pt);
    \end{tikzpicture}
}
\begin{figure}[t]
    \centering 
    \subfloat[CLF-CBF-QP - Ellipsoidal Fit]{
        \includegraphics[width=0.97\linewidth]{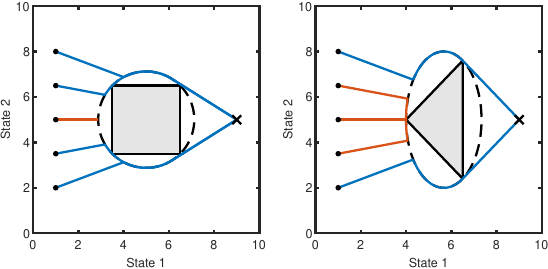}}\\
    \subfloat[CLF-CBF-QP - Polytopic Fit]{
        \includegraphics[width=0.97\linewidth]{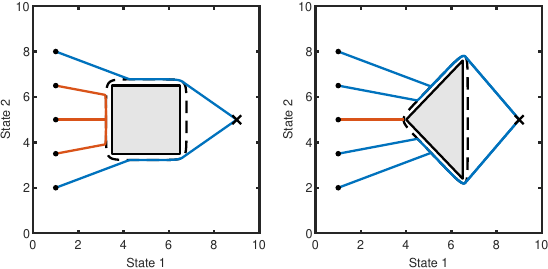}}\\
    \subfloat[MPC - Polytopic Fit]{
        \includegraphics[width=0.97\linewidth]{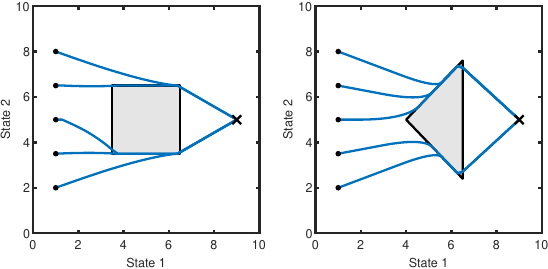}}
    \caption{System trajectories obtained using the three discussed strategies for two polytopic 
    unsafe sets. Blue trajectories indicate cases where the system successfully avoids the unsafe 
    set and reaches the desired equilibrium point. Meanwhile, orange trajectories denote cases in 
    which the system incurs in a deadlock situation. Dashed lines represent the boundary of the 
    safe set resulting from each approximation. The initial state is labeled as \xtarget and the 
    desired equilibrium point as \xchaserinit.}
    \label{Fig:IllustrativeExamples}
\end{figure}

Fig. \ref{Fig:IllustrativeExamples} (a) illustrates the CLF-CBF-QP approach, where an ellipsoid is 
used to approximate the unsafe set. In the examples shown in Fig. \ref{Fig:IllustrativeExamples}
(a), this ellipsoid corresponds to the minimum-volume ellipsoid that encloses the polytopic 
obstacle, obtained by solving a convex optimization problem \cite{cheung1993optimal}, 
\cite{van2009minimum}. However, while this approach enables a straightforward definition of a 
CBF, the ellipsoidal approximation is often too simplistic and may result in an overly 
conservative representation of certain complex unsafe regions. Additionally, as shown in Fig. 
\ref{Fig:IllustrativeExamples} (a), not all the trajectories reach the desired equilibrium point 
since an induced asymptotically stable equilibrium point appears at the boundary of the safe set, 
leading to deadlock situations.

Meanwhile, Fig. \ref{Fig:IllustrativeExamples} (b) illustrates the CLF-CBF-QP approach using a 
smooth over-approximation of the unsafe set, following the design from \eqref{Eq:SmoothMax}. As can 
be noticed, this strategy yields a closer fit to the actual unsafe set, where the conservativeness 
can be further reduced by increasing the value of $\kappa$. However, while this method has the 
potential to provide more accurate representations of more complex unsafe regions, it suffers 
from the same limitation as before: an induced asymptotically stable equilibrium point arises at 
the safe set boundary, preventing task completion for certain initial conditions.

Finally, Fig. \ref{Fig:IllustrativeExamples} (c) presents trajectories generated using the 
mixed-integer MPC approach, applied to directly account for the polytopic unsafe set. Similar to 
the strategy from Fig. \ref{Fig:IllustrativeExamples} (b), this method effectively models the 
unsafe set. However, with a sufficiently long prediction horizon, the MPC benefits from its 
predictive capability and avoids undesired equilibrium points, even in the symmetrical examples 
shown in Fig. \ref{Fig:IllustrativeExamples}, where there are two possible solutions to which the 
optimization solver can converge (either going above or below).

Nevertheless, despite its predictive advantages, the MPC approach requires solving a mixed-integer 
program at each sampling instant, adding a considerable computational demand for real-time 
applications. For instance, in the examples from Fig. \ref{Fig:IllustrativeExamples} (c), the 
Gurobi \cite{gurobi} solver was used with a sampling period of 0.1 seconds and a horizon of 20 
samples, achieving an average computation time of about 0.2 seconds. Additionally, since the MPC 
only enforces safety constraints at discrete-time instants, it struggles with navigating sharp 
corners, as can be noticed in Fig. \ref{Fig:IllustrativeExamples} (c). In contrast, the CLF-CBF-QP 
approach is highly computationally efficient and provides formal safety guarantees in continuous 
time, being well-suited for real-time, safety-critical applications. This discussion motivates the 
seek for an improved CLF-CBF-based approach that is capable of achieving decisiveness and avoid 
deadlocks.


\subsection{Problem Statement}

Motivated by the previous discussion, we now formally state the problem addressed in this paper.

\subsubsection*{Problem 1}

For $n \geq 2$, consider a first-order control-affine system defined as in 
\eqref{Eq:ControlAffineSystem}, such that $\mathbf{G}(\mathbf{x})$ has full row rank for all 
$\mathbf{x} \in \mathbb{R}^n$. Furthermore, let $\mathcal{O} \subset \mathbb{R}^n$ be a bounded 
unsafe set, and let $\bar{\mathbf{x}} \notin \mathcal{O}$ be a desired equilibrium point. Then, 
design a closed-form control strategy that renders a safe set 
$\mathcal{C} \subseteq \mathbb{R}^n\setminus\mathcal{O}$ forward invariant and $\bar{\mathbf{x}}$ an
asymptotically stable equilibrium point with region of attraction including $\mathcal{C}$. 

\subsubsection*{Problem 2} 

Extend the solution proposed for Problem 1 to higher-order control-affine systems of the 
form
\begin{equation}
    \begin{aligned}
        \Dot{\mathbf{x}} &= \mathbf{f}(\mathbf{x}) + \mathbf{G}(\mathbf{x})\mathbf{z}_1,\\
        \Dot{\mathbf{z}}_1 &= \mathbf{f}_1(\mathbf{x}, \bm{\eta}_1) 
        + \mathbf{G}_1(\mathbf{x}, \bm{\eta}_1)\mathbf{z}_2,\\
        &\,\,\,\vdots\\
        \Dot{\mathbf{z}}_r &= \mathbf{f}_r(\mathbf{x}, \bm{\eta}_r) 
        + \mathbf{G}_r(\mathbf{x}, \bm{\eta}_r)\mathbf{u},
    \end{aligned}
    \label{Eq:StrictFeedbackSystem}
\end{equation}
where, for each $i \in \{1, \dots, r\}$, the state $\mathbf{z}_i$ has dimension $p_i$, 
the auxiliary state $\bm{\eta}_i \in \mathbb{R}^{\ell_i}$ is defined as 
$\bm{\eta}_i = (\mathbf{z}_1, \dots, \mathbf{z}_i)$, the fields 
$\mathbf{f}_i: \mathbb{R}^n\times\mathbb{R}^{\ell_i} \rightarrow \mathbb{R}^{p_i}$ and 
$\mathbf{G}_i: \mathbb{R}^n\times\mathbb{R}^{\ell_i} \rightarrow \mathbb{R}^{p_i\times p_{i+1}}$ 
are locally Lipschitz continuous ($p_{r+1} = m$), and the gain matrix 
$\mathbf{G}_i(\mathbf{x}, \bm{\eta}_i)$ has full row rank for all 
$(\mathbf{x}, \bm{\eta}_i) \in \mathbb{R}^n\times\mathbb{R}^{\ell_i}$.

Note that the control-affine system defined in \eqref{Eq:StrictFeedbackSystem} extends the 
system from Problem 1 by concatenating it with additional control-affine systems in a cascaded 
manner, where only the lowest-level input can be directly controlled. Such systems are called the
strict-feedback systems. Also, we highlight that the system classes considered in Problems 1 and 2, 
in which all gain matrices have full row rank (fully actuated), are not overly conservative. In 
fact, they cover a broad range of applications, including robotic manipulators 
\cite{ferrara2000control}, ground, marine, and aerial vehicles \cite{jiangdagger1997tracking}, 
\cite{yang2004combined}, \cite{kim2004robust}, spacecraft \cite{farrell2005backstepping}, 
\cite{sun2016adaptive}, and so on.

In addition, even when the original system does not conform to these structures, it can often be 
reformulated to do so via a suitable coordinate transformation. For instance, in the case of the 
unicycle model, one can control a point slightly ahead from the original control point, which 
yields a first-order control-affine system with a full-row-rank gain matrix 
\cite{glotfelter2019hybrid}. This trick is actually very commonly used for controlling 
underactuated autonomous vehicles \cite{aguiar2007trajectory}, \cite{reis2022nonlinear}. Also, 
under certain conditions, more general systems can be transformed into a strict-feedback form with 
full-row-rank gain matrices \cite{krstivc1995nonlinear}, \cite{duan2021high}. This further 
underscores the generality of the systems under consideration.


\section{Hybrid Control Solution} \label{Sec:HybridController}

This section describes our proposed solution for first-order control-affine systems, as defined in 
Problem 1. To address this problem, we propose a hybrid control strategy that is based on a 
polytopic approximation of the unsafe set. More specifically, we consider a bounded convex polytope 
$\mathcal{P} \subset \mathbb{R}^n$ that encloses the unsafe set $\mathcal{O}$, so that 
$\mathcal{O} \subseteq \mathrm{int}(\mathcal{P})$ and 
$\bar{\mathbf{x}} \notin \mathrm{int}(\mathcal{P})$. The set $\mathcal{P}$ is defined by the 
intersection of $Q \geq n+1$ half-spaces:
\begin{equation}
    \mathcal{P} = \{\mathbf{x} \in \mathbb{R}^n: 
    h_1(\mathbf{x}) \leq 0  \wedge \dots \wedge h_Q(\mathbf{x}) \leq 0\},
    \label{Eq:PolytopeDefinition}
\end{equation}
where, for each $q \in \{1, \dots, Q\}$, $h_q: \mathbb{R}^n \rightarrow \mathbb{R}$ is defined by
\begin{equation}
    h_q(\mathbf{x}) = \mathbf{n}_q^\top\mathbf{x} - d_q.
    \label{Eq:CBFHalfspace}
\end{equation}
Here, the vector $\mathbf{n}_q \in \mathbb{S}^{n-1}$ denotes the unit normal associated with the 
facet $q$, pointing outward from $\mathcal{P}$, and $d_q \in \mathbb{R}$ is the respective offset. 
Additionally, we consider that the half-space constraints defining the polytope are nonredundant, 
meaning that removing any of the constraints would result in a distinct set. This implies that 
$\mathbf{n}_{q_1} \neq \mathbf{n}_{q_2}$ for every $q_1, q_2 \in \{1, \dots, Q\}$ such that 
$q_1 \neq q_2$. Based on this formulation, we define a safe set $\mathcal{C}$ as the closure of the 
complement of $\mathcal{P}$, i.e.,
\begin{equation}
    \mathcal{C} = \{\mathbf{x} \in \mathbb{R}^n: 
    h_1(\mathbf{x}) \geq 0 \vee \dots \vee h_Q(\mathbf{x}) \geq 0\},
    \label{Eq:HybridSafeSet}
\end{equation}
which guarantees that $\mathcal{C} \subseteq \mathbb{R}^n\setminus\mathcal{O}$ and 
$\bar{\mathbf{x}} \in \mathcal{C}$. 

\newpage

The hybrid control strategy relies on the safe set from \eqref{Eq:HybridSafeSet} and involves 
solving a series of safe stabilization subproblems. Each subproblem consists of an active safe 
half-space and an associated target point, for which we design a controller using compatible CLF 
and CBF constraints, ensuring convergence to the active setpoint. As the system approaches the 
active target, a switching mechanism updates the active safe half-space and establishes the new 
setpoint. Global asymptotic stabilization to the goal is then achieved by guaranteeing that the 
switching logic produces a setpoint sequence that converges to $\bar{\mathbf{x}}$.

The proposed control strategy can thus be described through an auxiliary hybrid system with flow 
given by
\begin{equation}
    \begin{bmatrix}
        \Dot{\mathbf{x}}\\
        \Dot{\bm{\xi}}
    \end{bmatrix} =
    \begin{bmatrix}
        \mathbf{f}(\mathbf{x}) + \mathbf{G}(\mathbf{x})\mathbf{k}_{\bm{\xi}}(\mathbf{x})\\
        \mathbf{0}
    \end{bmatrix},\quad (\mathbf{x}, \bm{\xi}) \in \mathcal{F},
    \label{Eq:HybridSystemFlow}
\end{equation}
and with jump dynamics described by
\begin{equation}
    \begin{bmatrix}
        \mathbf{x}^+\\
        \bm{\xi}^+
    \end{bmatrix} =
    \begin{bmatrix}
        \mathbf{x}\\
        \mathbf{s}(\mathbf{x}, \bm{\xi})
    \end{bmatrix},\quad (\mathbf{x}, \bm{\xi}) \in \mathcal{J}.
    \label{Eq:HybridSystemJump}
\end{equation}
Here, $\bm{\xi} = (\hat{\mathbf{x}}, q) \in \Xi \subset \mathbb{R}^n\times\{1, \dots, Q\}$ is an 
auxiliary state that includes the current target point, $\hat{\mathbf{x}}$, and the index $q$ of 
the active safe half-space, where the set $\Xi$ is defined as 
\begin{equation}
    \Xi = \{(\hat{\mathbf{x}}, q) \in \mathbb{R}^n\times\{1, \dots, Q\}: 
    h_{q}(\hat{\mathbf{x}}) \geq 0\}.
\end{equation}
Furthermore, $\mathbf{k}_{\bm{\xi}}: \mathbb{R}^n \rightarrow \mathbb{R}^m$ is a locally Lipschitz 
controller that solves the subproblem defined by the state $\bm{\xi}$. Specifically, 
$\mathbf{k}_{\bm{\xi}}$ makes the safe half-space $q$ forward invariant and the target 
$\hat{\mathbf{x}}$ asymptotically stable with region of attraction including the safe half-space 
$q$. As detailed later, the design of $\mathbf{k}_{\bm{\xi}}$ is based on compatible CLF and CBF 
constraints, which becomes feasible due to the linear nature of a half-space. The switching logic 
is determined by the flow and jump sets, 
$\mathcal{F}, \mathcal{J} \subset \mathcal{H} = \mathbb{R}^n\times\Xi$, along with the function 
$\mathbf{s}: \mathcal{J} \rightarrow \Xi$, which updates $\bm{\xi}$.

In what follows, we describe the proposed switching logic and the design of 
$\mathbf{k}_{\bm{\xi}}$. The objective is to ensure that, for every initial condition 
$\mathbf{x}_0 \in \mathcal{C}$, the hybrid control strategy produces a piecewise continuously 
differentiable solution $\bm{\varphi}: \mathbb{R}_{\geq 0} \rightarrow \mathbb{R}^n$, 
characterized by $K \in \mathbb{N}$ jumps as
\begin{equation}
    \bm{\varphi}(t) = 
    \begin{cases}
        \bm{\varphi}_0(t), &\text{if } t \in [0, t_1),\\
        \quad\vdots\\
        \bm{\varphi}_K(t), &\text{if } t \in [t_K, \infty),
    \end{cases}
\end{equation}
so that $\bm{\varphi}(t) \in \mathcal{C}$ for all $t \in \mathbb{R}_{\geq 0}$ and 
$\lim_{t\rightarrow\infty} \|\bm{\varphi}(t) - \bar{\mathbf{x}}\| = 0$.


\subsection{Switching Logic}

The proposed switching mechanism is based on a reference direction that serves to guide the 
sequence of setpoints toward the desired equilibrium point $\bar{\mathbf{x}}$. This direction is 
given by the vector $\mathbf{v} \in \mathbb{S}^{n-1}$, selected as the unit inward normal to 
a safe half-space containing $\bar{\mathbf{x}}$. Particularly, we define $\mathbf{v}$ as
\begin{equation}
    \mathbf{v} = \mathbf{n}_{\bar{q}},
    \label{Eq:ReferenceDirection}
\end{equation}
where $\bar{q}$ denotes the index of the function $h_q$ that achieves the highest value at 
$\bar{\mathbf{x}}$, meaning that\footnote{We consider that $\arg\max/\min$ returns a single 
solution. If the optimization problem has multiple solutions, $\arg\max/\min$ returns one of 
them.}
\begin{equation}
    \bar{q} = \underset{q' \in \{1, \dots, Q\}}{\arg\max} h_{q'}(\bar{\mathbf{x}}).
    \label{Eq:ReferenceHalfspace}
\end{equation}

Additionally, for a given hybrid mode defined by an active safe half-space $q$ and an active target 
point $\hat{\mathbf{x}}$, we define an auxiliary safe half-space $\hat{q}$, which forecasts the 
safe half-space that will be activated at an upcoming jump event. The auxiliary safe half-space 
$\hat{q}$ is defined through
\begin{equation}
    \hat{q} = \underset{q' \in \hat{\mathcal{Q}}_q}{\arg\max}\, h_{q'}(\hat{\mathbf{x}}),
    \label{Eq:ForecastHalfspace}
\end{equation}
where $\hat{\mathcal{Q}}_q$ is a subset of half-space indices defined as
\begin{equation}
    \hat{\mathcal{Q}}_q = \left\{q' \in \{1, \dots, Q\}: 
    \mathbf{v}^\top\mathbf{n}_{q'} > \mathbf{v}^\top\mathbf{n}_q\right\} \cup \{\bar{q}\}.
    \label{Eq:PredictionSet}
\end{equation}
This construction ensures that the safe half-space $\hat{q}$ is selected from safe half-spaces
whose normal vectors are more aligned with the reference normal $\mathbf{v} = \mathbf{n}_{\bar{q}}$
than the active one, and the union with $\{\bar{q}\}$ guarantees that the set $\hat{\mathcal{Q}}_q$ 
is nonempty when the active safe half-space is already $\bar{q}$.

Building on the previous auxiliary variables, the switching logic is based on a hysteretic 
behavior, inspired by synergistic Lyapunov functions \cite{mayhew2011synergistic}, 
\cite{mayhew2011further}. Specifically, for an active safe half-space $q$ that does not contain the 
desired equilibrium point $\bar{\mathbf{x}}$, we will enforce the active setpoint to be placed on 
the active boundary hyperplane, i.e., enforcing that
\begin{equation}
    h_q(\hat{\mathbf{x}}) = 0.
    \label{Eq:TargetPointBoundaryHyperplane}
\end{equation}
Moreover, we consider a minimum synergy gap $\mu \in \mathbb{R}_{>0}$, used as a reference for 
placing the intermediate setpoint relative to the safe half-space $\hat{q}$. Specifically, 
in addition to \eqref{Eq:TargetPointBoundaryHyperplane}, we will also enforce the intermediate 
target point to verify
\begin{equation}
    h_{\hat{q}}(\hat{\mathbf{x}}) \geq \mu.
    \label{Eq:TargetPointSynergyGap}
\end{equation}
This ensures that the intermediate setpoint lies in the interior of the safe half-space $\hat{q}$, 
with $\mu$ being a design parameter defining the minimum allowed distance to the boundary 
hyperplane $\hat{q}$.

Given this construction, and introducing a desired hysteresis width $\sigma \in (0, \mu)$, we 
define the flow and jump sets as follows:
\begin{equation}
    \begin{aligned}
        \mathcal{F} = \{(\mathbf{x}, \bm{\xi}) \in \mathcal{H}: 
        h_{\hat{q}}(\mathbf{x}) - h_q(\mathbf{x}) < \sigma \vee h_q(\mathbf{x}) < 0\},\\
        \mathcal{J} = \{(\mathbf{x}, \bm{\xi}) \in \mathcal{H}: 
        h_{\hat{q}}(\mathbf{x}) - h_q(\mathbf{x}) \geq \sigma \wedge h_q(\mathbf{x}) \geq 0\},
    \end{aligned}
    \label{Eq:JumpAndFlowSets}
\end{equation}
and the active safe half-space is updated according to
\begin{equation}
    q^+ = \hat{q},\quad\text{if } (\mathbf{x}, \bm{\xi}) \in \mathcal{J}.
    \label{Eq:UpdateActiveHalfplane}
\end{equation}
This formulation guarantees that safety is maintained when a jump occurs, as from the jump set 
definition we conclude that $h_{q^+}(\mathbf{x}) \geq h_q(\mathbf{x}) + \sigma > 0$. It also 
ensures that $\mathbf{v}^\top\mathbf{n}_{q^+}>\mathbf{v}^\top\mathbf{n}_q$, which means that the 
alignment of the active normal with the reference one increases with each jump. This latter 
condition plays a significant role, as it ensures that each half-space may only be active at most 
once, preventing the occurrence of limit cycles. It now remains to explain how the intermediate 
target points are computed to enforce the conditions in \eqref{Eq:TargetPointBoundaryHyperplane} 
and \eqref{Eq:TargetPointSynergyGap}.

To compute the intermediate target points, we assign an auxiliary tangent direction to each 
boundary hyperplane other than the reference one ($\bar{q}$). Each tangent direction is defined by 
an auxiliary vector $\mathbf{t}_q$, obtained by projecting the reference vector $\mathbf{v}$ onto 
the respective hyperplane. Specifically, for each $q \in \{1, \dots, Q\}\setminus\{\bar{q}\}$, the 
tangent vector $\mathbf{t}_q$ is given by
\begin{equation}
    \mathbf{t}_q = \left(\mathbf{I} - \mathbf{n}_q\mathbf{n}_q^\top\right)\mathbf{v} 
    + \bm{\epsilon}\chi
    \left(\left\|\left(\mathbf{I} - \mathbf{n}_q\mathbf{n}_q^\top\right)\mathbf{v}\right\|\right),
    \label{Eq:ReferenceDirectionProjection}
\end{equation}
where $\chi: \mathbb{R} \rightarrow \{0, 1\}$ is an indicator function defined as

\newpage

\begin{equation}
    \chi(s) = 
    \begin{cases}
        1, &\text{if } s = 0,\\
        0, &\text{if } s \neq 0,
    \end{cases}
    \label{Eq:DecisivenessDirection}
\end{equation}
and the vector $\bm{\epsilon} \neq \mathbf{0}$ is such that $\mathbf{v}^\top\bm{\epsilon} = 0$. The 
first term in \eqref{Eq:ReferenceDirectionProjection} is the projection of $\mathbf{v}$ onto the 
hyperplane defined by the normal vector $\mathbf{n}_q$, and the second term guarantees decisiveness 
when $\mathbf{n}_q = -\mathbf{v}$. In this paper, we do not adhere to a particular method for 
choosing $\bm{\epsilon}$, but a direct approach is to select a fixed direction arbitrarily or draw 
it from a probability distribution. Alternatively, more optimized approaches may be explored.

Consider now a hybrid mode defined by the auxiliary state $\bm{\xi} = (\hat{\mathbf{x}}, q)$, under 
the assumption that the conditions \eqref{Eq:TargetPointBoundaryHyperplane} and 
\eqref{Eq:TargetPointSynergyGap} are satisfied. At the next jump event, we assess whether the 
forecast safe half-space $\hat{q}$ (which becomes active after the jump) contains the desired 
equilibrium point $\bar{\mathbf{x}}$. If it does, the next setpoint is simply set to 
$\bar{\mathbf{x}}$, being this the simplest case. Conversely, when the safe half-space $\hat{q}$ 
does not contain $\bar{\mathbf{x}}$, we proceed as follows. First, we determine the point 
$\tilde{\mathbf{x}}$ where the boundary hyperplane $\hat{q}$ intersects the line segment between 
the current state $\mathbf{x}$ (at the time of the jump) and $\bar{\mathbf{x}}$. We then adjust 
this point by adding the tangent vector $\mathbf{t}_{\hat{q}}$ scaled by a factor 
$\tau \in \mathbb{R}_{\geq 0}$, resulting in the new target point. The scaling factor $\tau$ is 
chosen as the smallest nonnegative value that ensures the new target point also satisfies the 
synergy gap condition, and it is obtained through an auxiliary optimization problem.

More specifically, the active target point is updated as
\begin{equation}
    \hat{\mathbf{x}}^+ = 
    \begin{cases}
        \tilde{\mathbf{x}} + \mathbf{t}_{\hat{q}}\tau,
        &\text{if } h_{\hat{q}}(\bar{\mathbf{x}}) < 0 
        \text{ and } (\mathbf{x}, \bm{\xi}) \in \mathcal{J},\\
        \bar{\mathbf{x}}, &\text{if } h_{\hat{q}}(\bar{\mathbf{x}}) \geq 0
        \text{ and } (\mathbf{x}, \bm{\xi}) \in \mathcal{J},
    \end{cases}
    \label{Eq:UpdateTargetPoint}
\end{equation}
where $\tilde{\mathbf{x}}$ is the intersection point between the line segment 
$\overline{\mathbf{x}\bar{\mathbf{x}}}$ and the boundary hyperplane $\hat{q}$, which can be 
computed as
\begin{equation}
    \tilde{\mathbf{x}} = \mathbf{x} 
    + \frac{h_{\hat{q}}(\mathbf{x})}{h_{\hat{q}}(\mathbf{x}) - h_{\hat{q}}(\bar{\mathbf{x}})}
    (\bar{\mathbf{x}} - \mathbf{x}),
    \label{Eq:IntersectionPoint}
\end{equation}
and the scaling factor $\tau \in \mathbb{R}_{\geq 0}$ is determined as follows:
\begin{equation} 
    \begin{aligned}
        \tau =\,\, &\min_{(\tau', q') \in \mathbb{R}_{\geq 0}\times\hat{\mathcal{Q}}_{\hat{q}}} 
        \tau'\\
        \text{subject to}\,\, &h_{q'}(\tilde{\mathbf{x}} + \mathbf{t}_{\hat{q}}\tau') \geq \mu.
        \label{Eq:ScalingFactorOP}
    \end{aligned}
\end{equation}
The optimization problem in \eqref{Eq:ScalingFactorOP} can be compactly written as
\begin{equation} 
    \begin{aligned}
        \tau = \min_{q' \in \hat{\mathcal{Q}}_{\hat{q}}} \tau_{q'},
    \end{aligned}
    \label{Eq:ScalingFactor}
\end{equation}
where, for each $q' \in \hat{\mathcal{Q}}_{\hat{q}}$, $\tau_{q'}$ is given by
\begin{equation}
    \scalebox{0.89}{$
    \tau_{q'} = 
    \begin{cases}
        \left(\mathbf{n}_{q'}^\top\mathbf{t}_{\hat{q}}\right)^{-1}
        (\mu - h_{q'}(\tilde{\mathbf{x}})), 
        &\text{if } h_{q'}(\tilde{\mathbf{x}}) < \mu,\, 
        \mathbf{n}_{q'}^\top\mathbf{t}_{\hat{q}} > 0,\\
        0, &\text{if } h_{q'}(\tilde{\mathbf{x}}) \geq \mu,\\
        \infty\,\, (\text{infeasible}), &\text{otherwise}.
    \end{cases}$}
    \label{Eq:ScalingFactorAux}
\end{equation}

Note that when the safe half-space $\hat{q}$ does not contain $\bar{\mathbf{x}}$, the next target 
point is placed on the boundary hyperplane $\hat{q}$, and the value of $\tau$ obtained through 
\eqref{Eq:ScalingFactor}-\eqref{Eq:ScalingFactorAux} ensures that the next setpoint also satisfies 
the synergy gap condition. In particular, note that in \eqref{Eq:ScalingFactor}, we consider the 
set $\hat{\mathcal{Q}}_{\hat{q}}$, which corresponds to the safe half-spaces whose normal vectors 
are more aligned with the reference vector $\mathbf{v}$ than that of the half-space $\hat{q}$, as 
defined in \eqref{Eq:PredictionSet}. Therefore, if the current hybrid mode satisfies conditions
\eqref{Eq:TargetPointBoundaryHyperplane} and \eqref{Eq:TargetPointSynergyGap}, the update rules from
\eqref{Eq:UpdateActiveHalfplane} and \eqref{Eq:UpdateTargetPoint} ensure that the next mode also 
satisfies these conditions when $h_{\hat{q}}(\bar{\mathbf{x}}) < 0$. The switching mechanism is 
illustrated in Fig. \ref{Fig:SwitchingLogic}.

\newpage

\begin{figure}[t] 
    \centering 
    \includegraphics[width=\linewidth]{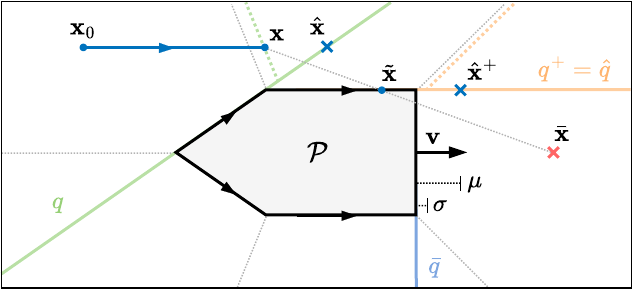} 
    \caption{This figure illustrates the switching mechanism triggered by a jump event when the 
    forecast safe half-space $\hat{q}$ does not contain the desired equilibrium point 
    $\bar{\mathbf{x}}$. In the figure, the active safe half-space $q$ is the one containing the 
    hypothetical blue trajectory and whose boundary is the diagonal green line on which the current 
    target point $\hat{\mathbf{x}}$ lies. The auxiliary safe half-space $\hat{q}$ (which will be 
    the active one after the jump) is the top horizontal one, whose boundary is the top horizontal 
    orange line (only partially represented). The current target point verifies the synergy gap 
    condition $h_{\hat{q}}(\hat{\mathbf{x}}) \geq \mu$. Also, the reference safe half-space 
    $\bar{q}$ is the one defined by the reference normal vector 
    $\mathbf{v} = \mathbf{n}_{\bar{q}}$, which contains $\bar{\mathbf{x}}$ and whose boundary is 
    the blue vertical line (only partially represented). The arrows on each facet of the polytope 
    illustrate the auxiliary tangent directions defined the vectors $\mathbf{t}_q$, obtained by 
    projecting the reference direction defined by $\mathbf{v}$ on each boundary hyperplane. 
    \textbf{Switching Logic:} The figure depicts the exact moment a jump occurs, where the current 
    state $\mathbf{x}$ reaches the boundary of the current jump region. This boundary is defined by 
    the conditions $h_{\hat{q}}(\mathbf{x}) - h_q(\mathbf{x}) = \sigma$ and 
    $h_q(\mathbf{x}) \geq 0$, and it is represented by the green dotted line on which $\mathbf{x}$ 
    lies. Since the safe half-space $\hat{q}$ does not contain $\bar{\mathbf{x}}$, the next 
    setpoint is computed according to the first case in \eqref{Eq:UpdateTargetPoint}. Specifically, 
    we locate the point $\tilde{\mathbf{x}}$ where the boundary hyperplane $\hat{q}$ intersects the 
    line segment $\overline{\mathbf{x}\bar{\mathbf{x}}}$. We then shift $\tilde{\mathbf{x}}$ in the 
    direction of the tangent vector $\mathbf{t}_{\hat{q}}$ until the synergy gap condition is 
    verified, resulting in the new setpoint $\hat{\mathbf{x}}^+$. In this case, this corresponds to 
    shifting $\tilde{\mathbf{x}}$ so that $h_{\bar{q}}(\hat{\mathbf{x}}^+) = \mu$.}
    \label{Fig:SwitchingLogic} 
\end{figure} 

Finally, it now only remains to ensure that the initial hybrid mode satisfies conditions 
\eqref{Eq:TargetPointBoundaryHyperplane} and \eqref{Eq:TargetPointSynergyGap} when the initial 
active safe half-space does not contain $\bar{\mathbf{x}}$. To address this, the initial auxiliary 
state $\bm{\xi}_0 = (\hat{\mathbf{x}}_0, q_0)$ is computed through a procedure that serves as a 
pre-initial update. Specifically, the active safe half-space is initialized through
\begin{equation}
    q_0 = \underset{q' \in \{1, \dots, Q\}}{\arg\max} h_{q'}(\mathbf{x}_0),
    \label{Eq:InitializationHalfSpace}
\end{equation}
and the initial target point is determined according to
\begin{equation}
    \hat{\mathbf{x}}_0 = 
    \begin{cases}
        \tilde{\mathbf{x}}_0 + \mathbf{t}_{q_0}\tau_0, 
        &\text{if } h_{q_0}(\bar{\mathbf{x}}) < 0,\\
        \bar{\mathbf{x}}, &\text{if } h_{q_0}(\bar{\mathbf{x}}) \geq 0,
    \end{cases}
    \label{Eq:InitializationTargetPoint}
\end{equation}
where $\tilde{\mathbf{x}}_0$ is the point of intersection between the line segment 
$\overline{\mathbf{x}_0\bar{\mathbf{x}}}$ and the boundary hyperplane $q_0$, and $\tau_0$ is 
computed using \eqref{Eq:ScalingFactor} with $\hat{q} = q_0$ and 
$\tilde{\mathbf{x}} = \tilde{\mathbf{x}}_0$. 

With the proposed switching mechanism now fully outlined, Algorithm \ref{Alg:hybrid} provides a 
concise overview of the hybrid control strategy, highlighting its implementation from a 
computational standpoint. We are now ready to present the main result of this paper, which is 
stated in Theorem \ref{Th:MyTheorem1} and is supported by the auxiliary result from Lemma 
\ref{Th:MyLemma1}.

\begin{algorithm}[t]
    \caption{Hybrid Control Algorithm}
    \label{Alg:hybrid}
    \begin{algorithmic}[1]
        \Require $\mathbf{x}_0$, $\bar{\mathbf{x}}$, $\mathbf{n}_1, \dots, \mathbf{n}_Q$, $d_1, 
        \dots, d_Q$, $\mu$, $\sigma$, $\mathbf{k}_{\bm{\xi}}$
        \State Compute $\mathbf{v}$ using 
        \eqref{Eq:ReferenceDirection}-\eqref{Eq:ReferenceHalfspace}
        \Comment{Setup}
        \State Initialize $\bm{\xi} = (\hat{\mathbf{x}}, q)$ using 
        \eqref{Eq:InitializationHalfSpace}-\eqref{Eq:InitializationTargetPoint}
        \State Initialize $\hat{q}$ using \eqref{Eq:ForecastHalfspace}-\eqref{Eq:PredictionSet}
        \While{true} \Comment{Control Loop}
            \State Get current state $\mathbf{x}$
            \If{$h_{\hat{q}}(\mathbf{x}) - h_q(\mathbf{x}) \geq \sigma$ and 
            $h_q(\mathbf{x}) \geq 0$} \Comment{Jump}
                \If{$h_{\hat{q}}(\bar{\mathbf{x}}) < 0$}
                    \State Compute $\tilde{\mathbf{x}}$, $\mathbf{t}_{\hat{q}}$, $\tau$ via 
                    \eqref{Eq:IntersectionPoint}, \eqref{Eq:ReferenceDirectionProjection}, 
                    \eqref{Eq:ScalingFactor}-\eqref{Eq:ScalingFactorAux}
                    \State $(\hat{\mathbf{x}}, q) \gets 
                    (\tilde{\mathbf{x}} + \mathbf{t}_{\hat{q}}\tau, \hat{q})$
                \Else
                    \State $(\hat{\mathbf{x}}, q) \gets (\bar{\mathbf{x}}, \hat{q})$
                \EndIf
                \State Update $\hat{q}$ using \eqref{Eq:ForecastHalfspace}-\eqref{Eq:PredictionSet}
            \EndIf
            \State Apply input $\mathbf{k}_{\bm{\xi}}(\mathbf{x})$ to the system
        \EndWhile
    \end{algorithmic}
\end{algorithm}

\begin{lemma} \label{Th:MyLemma1}
    For the nonlinear control-affine system \eqref{Eq:ControlAffineSystem}, consider the safe 
    stabilization problem associated with the auxiliary variable 
    $\bm{\xi} = (\hat{\mathbf{x}}, q) \in \Xi$, and assume there exists a locally Lipschitz 
    continuous controller 
    $\mathbf{k}_{\bm{\xi}}: \mathbb{R}^n \rightarrow \mathbb{R}^m$ that solves this problem, which 
    produces a solution $\bm{\varphi}_{\bm{\xi}}: \mathbb{R}_{\geq 0} \rightarrow \mathbb{R}^n$ 
    when applied to \eqref{Eq:ControlAffineSystem}. If the target point $\hat{\mathbf{x}}$ satisfies
    $h_{\hat{q}}(\hat{\mathbf{x}}) - h_q(\hat{\mathbf{x}}) > \sigma$, then, for every initial 
    condition $\mathbf{x}_0$ such that $h_q(\mathbf{x}_0) \geq 0$, there exists a finite time 
    instant $t_s \in \mathbb{R}_{\geq 0}$ in which 
    $(\bm{\varphi}_{\bm{\xi}}(t_s), \bm{\xi}) \in \mathcal{J}$.
\end{lemma}

\begin{proof}
    Let $\tilde{\mathcal{J}}_{\bm{\xi}} = \{\mathbf{x} \in \mathbb{R}^n: 
    h_{\hat{q}}(\mathbf{x}) - h_q(\mathbf{x}) \geq \sigma\}$. If 
    $h_{\hat{q}}(\hat{\mathbf{x}}) - h_q(\hat{\mathbf{x}}) > \sigma$, then 
    $\hat{\mathbf{x}} \in \mathrm{int}\big(\tilde{\mathcal{J}}_{\bm{\xi}}\big) \neq \emptyset$. 
    Furthermore, for every initial condition $\mathbf{x}_0$ so that $h_q(\mathbf{x}_0) \geq 0$, 
    $\mathbf{k}_{\bm{\xi}}$ ensures that $h_q(\bm{\varphi}_{\bm{\xi}}(t)) \geq 0$ for all 
    $t \in \mathbb{R}_{\geq 0}$ and 
    $\lim_{t\rightarrow\infty} \|\bm{\varphi}_{\bm{\xi}}(t) - \hat{\mathbf{x}}\| = 0$. 
    Therefore, as $\hat{\mathbf{x}}$ belongs to the interior of $\tilde{\mathcal{J}}_{\bm{\xi}}$, 
    this result follows directly from the definition of limit of a function.
\end{proof}

\begin{theorem} \label{Th:MyTheorem1}
    Consider the control-affine system \eqref{Eq:ControlAffineSystem}. If there exists a 
    locally Lipschitz continuous controller 
    $\mathbf{k}_{\bm{\xi}}: \mathbb{R}^n \rightarrow \mathbb{R}^m$ that solves the safe 
    stabilization subproblem defined by the auxiliary state 
    $\bm{\xi} = (\hat{\mathbf{x}}, q) \in \Xi$, then the hybrid control strategy described by 
    \eqref{Eq:PolytopeDefinition}-\eqref{Eq:InitializationTargetPoint} renders the safe set 
    $\mathcal{C}$ in \eqref{Eq:HybridSafeSet} forward invariant and 
    $\bar{\mathbf{x}} \in \mathcal{C}$ an asymptotically stable equilibrium point with region of 
    attraction including $\mathcal{C}$.
\end{theorem}

\begin{proof}
    To prove this result, we begin by noting that safety is maintained during flows and jumps. 
    During flows, safety is guaranteed by $\mathbf{k}_{\bm{\xi}}$, and during jumps, safety is 
    maintained since by \eqref{Eq:JumpAndFlowSets} and \eqref{Eq:UpdateActiveHalfplane} we have 
    that $h_{q^+}(\mathbf{x}) \geq h_q(\mathbf{x}) + \sigma > 0$. 

    Then, we note that when $\hat{\mathbf{x}} \neq \bar{\mathbf{x}}$, we necessarily have that 
    $h_q(\hat{\mathbf{x}}) = 0$ and $h_{\hat{q}}(\hat{\mathbf{x}}) \geq \mu$. This is ensured by 
    the initialization step \eqref{Eq:InitializationHalfSpace}-\eqref{Eq:InitializationTargetPoint} 
    and the update rule \eqref{Eq:UpdateActiveHalfplane}-\eqref{Eq:ScalingFactor}, which rely on 
    the optimization problem \eqref{Eq:ScalingFactorOP}, always solvable when 
    $h_{\hat{q}}(\bar{\mathbf{x}}) < 0$. Note that, when $\mathbf{n}_{\hat{q}} \neq -\mathbf{v}$, 
    the problem is feasible for at least $q' = \bar{q}$ since 
    $\bar{q} \in \hat{\mathcal{Q}}_{\hat{q}}$ and $\mathbf{v}^\top\mathbf{t}_{\hat{q}} > 0$. When 
    $\mathbf{n}_{\hat{q}} = -\mathbf{v}$, \eqref{Eq:ScalingFactorOP} is infeasible for 
    $q' = \bar{q}$, but since $n \geq 2$, there exists at least one 
    $q' \in \hat{\mathcal{Q}}_{\hat{q}}$ such that $\mathbf{n}_{q'}^\top\mathbf{t}_{\hat{q}} > 0$, 
    making the problem feasible. 

    By Lemma \ref{Th:MyLemma1}, we thus conclude that a jump will necessarily occur when 
    $\hat{\mathbf{x}} \neq \bar{\mathbf{x}}$ since 
    $h_{\hat{q}}(\hat{\mathbf{x}}) - h_q(\hat{\mathbf{x}}) = h_{\hat{q}}(\hat{\mathbf{x}}) 
    \geq \mu > \sigma$. Moreover, since 
    $\mathbf{v}^\top\mathbf{n}_{q^+}>\mathbf{v}^\top\mathbf{n}_q$ when a jump occurs, a safe 
    half-space will never be active more than once. Therefore, we conclude that a finite number of 
    jumps will occur until $\hat{\mathbf{x}} = \bar{\mathbf{x}}$.
    
    Once  we have $\hat{\mathbf{x}} = \bar{\mathbf{x}}$, one of two possible outcomes occurs. If 
    $h_{\bar{q}}(\bar{\mathbf{x}}) - h_q(\bar{\mathbf{x}}) \leq \sigma$, no further jumps occur and 
    the system converges to $\bar{\mathbf{x}}$. Meanwhile, if 
    $h_{\bar{q}}(\bar{\mathbf{x}}) - h_q(\bar{\mathbf{x}}) > \sigma$, one more jump occurs, after 
    which the system converges to $\bar{\mathbf{x}}$. In the second scenario, in the end we have 
    that $q = \bar{q}$, and no more jumps occur since 
    $h_{\bar{q}}(\bar{\mathbf{x}}) - h_q(\bar{\mathbf{x}}) = 0 \leq \sigma$.
\end{proof}


\subsection{Subproblem Controller Design} \label{Sec:SubproblemController}

To complete the description of the hybrid control strategy, it now remains to address the design of 
the controller $\mathbf{k}_{\bm{\xi}}$, which must solve the safe stabilization subproblem 
associated with the auxiliary state $\bm{\xi} = (\hat{\mathbf{x}}, q) \in \Xi$. Specifically, 
$\mathbf{k}_{\bm{\xi}}$ must render the safe half-space $q$ forward invariant and the target point 
$\hat{\mathbf{x}}$ asymptotically stable with its region of attraction including the active safe 
half-space $q$.

To this end, consider a CLF $V_{\hat{\mathbf{x}}}$ and a CBF $h_q$ for the system 
\eqref{Eq:ControlAffineSystem}, which define the following constraints for all 
$\mathbf{x} \in \mathbb{R}^n$:
\begin{equation}
    \begin{aligned}
        L_{\mathbf{f}}V_{\hat{\mathbf{x}}}(\mathbf{x}) 
        + L_{\mathbf{G}}V_{\hat{\mathbf{x}}}(\mathbf{x})\mathbf{u} 
        &\leq -\gamma(V_{\hat{\mathbf{x}}}(\mathbf{x})),\\
        L_{\mathbf{f}}h_q(\mathbf{x}) + L_{\mathbf{G}}h_q(\mathbf{x})\mathbf{u} 
        &\geq -\alpha(h_q(\mathbf{x})),
    \end{aligned}
\end{equation}
where $\gamma$ is a class-$\mathcal{K}$ function associated with the CLF and $\alpha$ is an 
extended class-$\mathcal{K}_\infty$ function corresponding to the CBF. To achieve both the 
stabilization and safety objectives, the CLF and CBF constraints must always be mutually 
satisfiable, that is, the intersection of the control half-spaces defined by these constraints must 
be nonempty for all $\mathbf{x} \in \mathbb{R}^n$. This intersection is always nonempty regardless 
of the values of $\gamma(V_{\hat{\mathbf{x}}}(\mathbf{x}))$ and $\alpha(h_q(\mathbf{x}))$, except 
when the vectors $L_{\mathbf{G}}V_{\hat{\mathbf{x}}}(\mathbf{x})$ and 
$L_{\mathbf{G}}h_q(\mathbf{x})$ are nonzero and aligned in the same direction. Such critical cases 
are captured by the set $\mathcal{S}_{\bm{\xi}}^\mathrm{c} \subset \mathbb{R}^n$, defined as
\begin{equation}
    \begin{aligned}
        \mathcal{S}_{\bm{\xi}}^\mathrm{c} = \{\mathbf{x} \in \mathbb{R}^n: 
        \exists \lambda \in \mathbb{R}_{> 0}:
        L_{\mathbf{G}}V_{\hat{\mathbf{x}}}(\mathbf{x}) &= \lambda L_{\mathbf{G}}h_q(\mathbf{x}),\\
        L_{\mathbf{G}}V_{\hat{\mathbf{x}}}(\mathbf{x}), L_{\mathbf{G}}h_q(\mathbf{x}) 
        &\neq \mathbf{0}\},
    \end{aligned}
    \label{Eq:CriticalSet1}
\end{equation} 
where the first condition in \eqref{Eq:CriticalSet1} can be expanded as
\begin{equation}
    (\nabla V_{\hat{\mathbf{x}}}(\mathbf{x}) 
    - \lambda \nabla h_q(\mathbf{x}))^\top\mathbf{G}(\mathbf{x}) = \mathbf{0}.
    \label{Eq:CriticalSet2}
\end{equation}
As $\mathbf{G}(\mathbf{x})$ is assumed to have full row rank, \eqref{Eq:CriticalSet2} simplifies to
\begin{equation}
    \nabla V_{\hat{\mathbf{x}}}(\mathbf{x}) = \lambda \nabla h_q(\mathbf{x}),
    \label{Eq:CriticalSet3}
\end{equation}
and thus we can also conclude that
\begin{equation}
    \lambda 
    = \frac{\|L_{\mathbf{G}}V_{\hat{\mathbf{x}}}(\mathbf{x})\|}{\|L_{\mathbf{G}}h_q(\mathbf{x})\|} 
    = \frac{\|\nabla V_{\hat{\mathbf{x}}}(\mathbf{x})\|}{\|\nabla h_q(\mathbf{x})\|}.
    \label{Eq:CriticalLambda}
\end{equation}
Now, by substituting \eqref{Eq:CriticalSet3} into the CLF constraint, we conclude that the CLF and 
CBF constraints become equivalent to
\begin{equation}
    -\alpha(h_q(\mathbf{x})) \leq 
    L_{\mathbf{f}}h_q(\mathbf{x}) + L_{\mathbf{G}}h_q(\mathbf{x})\mathbf{u} 
    \leq -\lambda^{-1}\gamma(V_{\hat{\mathbf{x}}}(\mathbf{x}))
\end{equation}
when $\mathbf{x} \in \mathcal{S}_{\bm{\xi}}^\mathrm{c}$. Hence, to ensure compatibility of the CLF 
and CBF constraints, $\gamma$ and $\alpha$ must be selected so that, for 
$\mathbf{x} \in \mathcal{S}_{\bm{\xi}}^\mathrm{c}$,
\begin{equation}
    \alpha(h_q(\mathbf{x})) \geq 
    \frac{\|\nabla h_q(\mathbf{x})\|}{\|\nabla V_{\hat{\mathbf{x}}}(\mathbf{x})\|}
    \gamma(V_{\hat{\mathbf{x}}}(\mathbf{x})).
    \label{Eq:CLF-CBF-ConditionsCritical}
\end{equation}

Consider now the typical choice of a quadratic CLF, where
\begin{equation}
    V_{\hat{\mathbf{x}}}(\mathbf{x}) = \frac{1}{2}\|\mathbf{x}-\hat{\mathbf{x}}\|^2
    \label{Eq:SubproblemCLF}
\end{equation}
for all $\mathbf{x} \in \mathbb{R}^n$, along with the CBF $h_q$ defined as in 
\eqref{Eq:CBFHalfspace}. The function $V_{\hat{\mathbf{x}}}$ defined in \eqref{Eq:SubproblemCLF} is 
a valid CLF for \eqref{Eq:ControlAffineSystem} when $\mathbf{G}(\mathbf{x})$ has full row rank, as 
for any class-$\mathcal{K}$ function $\gamma$ we have that
\begin{equation}
    \inf_{\mathbf{u} \in \mathbb{R}^m} 
    [L_\mathbf{f}V_{\hat{\mathbf{x}}}(\mathbf{x}) 
    + L_\mathbf{G}V_{\hat{\mathbf{x}}}(\mathbf{x})\mathbf{u}] = -\infty
    < - \gamma(V_{\hat{\mathbf{x}}}(\mathbf{x}))
\end{equation}
for all $\mathbf{x} \in \mathbb{R}^n\setminus\{\hat{\mathbf{x}}\}$. Similarly, the function $h_q$ 
defined as in \eqref{Eq:CBFHalfspace} is a valid CBF for \eqref{Eq:ControlAffineSystem} when 
$\mathbf{G}(\mathbf{x})$ has full row rank, as for any extended class-$\mathcal{K}_\infty$ 
function $\alpha$ it holds that
\begin{equation}
    \sup_{\mathbf{u} \in \mathbb{R}^m}
    [L_\mathbf{f}h_q(\mathbf{x}) + L_\mathbf{G}h_q(\mathbf{x})\mathbf{u}] = \infty 
    > - \alpha(h_q(\mathbf{x}))
\end{equation}
for all $\mathbf{x} \in \mathbb{R}^n$. Given these design choices, the critical set in 

\newpage

\noindent \eqref{Eq:CriticalSet1} becomes an open ray along the direction of $\mathbf{n}_q$, given 
by
\begin{equation}
    \mathcal{S}_{\bm{\xi}}^\mathrm{c} = \{\hat{\mathbf{x}} + \mathbf{n}_q \lambda: 
    \lambda \in \mathbb{R}_{> 0}\}.
\end{equation}
Additionally, if we choose $\gamma(s) = 2\bar{\gamma}s$ for all $s \in \mathbb{R}_{\geq 0}$ and 
$\alpha(s) = \bar{\alpha}s$ for all $s \in \mathbb{R}$, with $\bar{\gamma}, \bar{\alpha} \in 
\mathbb{R}_{> 0}$, \eqref{Eq:CLF-CBF-ConditionsCritical} simplifies to
\begin{equation}
    \begin{aligned}
        (\bar{\alpha} - \bar{\gamma})\lambda + \bar{\alpha}h_q(\hat{\mathbf{x}}) \geq 0
    \end{aligned}
\end{equation}
for all $\lambda \in \mathbb{R}_{> 0}$. Thus, with this direct approach, compatibility between the 
CLF and CBF constraints is achieved by selecting
\begin{equation}
    \bar{\alpha} \geq \bar{\gamma}.
\end{equation}
Nevertheless, alternative design choices may be considered, as long as the condition in 
\eqref{Eq:CLF-CBF-ConditionsCritical} holds for all 
$\mathbf{x} \in \mathcal{S}_{\bm{\xi}}^\mathrm{c}$.

Given compatible CLF and CBF constraints, we can then construct an optimization-based controller 
$\mathbf{k}_{\bm{\xi}}$ as follows:
\begin{equation} 
    \begin{aligned}
        \mathbf{k}_{\bm{\xi}}(\mathbf{x}) = 
        \,\, &\underset{\mathbf{u} \in \mathbb{R}^m}{\arg\min}\,\, 
        \frac{1}{2}\|\mathbf{u}\|^2\\
        \text{subject to}\,\, &L_{\mathbf{f}}V_{\hat{\mathbf{x}}}(\mathbf{x})
        + L_{\mathbf{G}}V_{\hat{\mathbf{x}}}(\mathbf{x})\mathbf{u} 
        \leq -\gamma(V_{\hat{\mathbf{x}}}(\mathbf{x})),\\
        &L_{\mathbf{f}}h_q(\mathbf{x}) + L_{\mathbf{G}}h_q(\mathbf{x})\mathbf{u} 
        \geq -\alpha(h_q(\mathbf{x})),
    \end{aligned}
    \label{Eq:CLF-CBF-QP-Compatible}
\end{equation}
which solves the subproblem defined by $\bm{\xi} = (\hat{\mathbf{x}}, q) \in \Xi$. If the CLF and 
CBF gradients, along with the functions $\gamma$ and $\alpha$, are locally Lipschitz continuous, 
and the inequality \eqref{Eq:CLF-CBF-ConditionsCritical} holds strictly, then the controller 
defined in \eqref{Eq:CLF-CBF-QP-Compatible} is locally Lipschitz continuous on 
$\mathbb{R}^n\setminus\{\hat{\mathbf{x}}\}$ \cite{morris2013sufficient}, \cite{jankovic2018robust}. 
Also, the controller can be expressed in closed form using 
\eqref{Eq:ControllerExpressionStart}-\eqref{Eq:ControllerExpressionEnd} with $p = \infty$.

\subsubsection*{Handling Input Bounds}

Until now, we have not considered input bounds because achieving global asymptotic stabilization 
for an arbitrary drift field $\mathbf{f}$ requires enough control authority to counteract the drift 
field. However, it is important to discuss how the proposed framework can accommodate input bounds, 
as real-world systems inherently operate under bounded inputs.

To account for input bounds within the proposed approach, we have to modify the QP formulation used 
in the subproblem controller design. For instance, consider an input bound of the form 
$\|\mathbf{u}\| \leq u_\mathrm{max}$, where $u_\mathrm{max} \in \mathbb{R}_{>0}$. To ensure that 
the CLF constraint is compatible with the input bound, we have to relax the stabilization 
objective, which can be done by introducing a relaxation function 
$\delta_{\hat{\mathbf{x}}}: \mathbb{R}^n \rightarrow \mathbb{R}_{\geq 0}$. In this setup, the 
``best effort'' we can do to achieve the stabilization goal is to relax the CLF constraint only as 
much as necessary to ensure that the respective control half-space intersects the admissible input 
set. This corresponds to defining $\delta_{\hat{\mathbf{x}}}$ for all $\mathbf{x} \in \mathbb{R}^n$ 
as
\begin{equation}
    \scalebox{0.91}{$
    \delta_{\hat{\mathbf{x}}}(\mathbf{x}) = \max\{0, L_{\mathbf{f}}V_{\hat{\mathbf{x}}}(\mathbf{x}) 
    + \gamma(V_{\hat{\mathbf{x}}}(\mathbf{x}))
    - \|L_{\mathbf{G}}V_{\hat{\mathbf{x}}}(\mathbf{x})\|u_\mathrm{max}\}.$}
\end{equation}

Then, to incorporate the safety objective, we can adopt the optimal-decay QP formulation introduced 
in \cite{zeng2021safety}, which leads to the following 
subproblem controller design:
\begin{align}
    (\mathbf{k}_{\bm{\xi}}(\mathbf{x}), \cdot)
    =\,\, &\underset{(\mathbf{u}, \omega) \in \mathbb{R}^{m+1}}{\arg\min}\,\,
    \frac{1}{2}\|\mathbf{u}\|^2 + \frac{1}{2}p(\omega - 1)^2\nonumber\\
    \text{subject to}\,\, &L_{\mathbf{f}}V_{\hat{\mathbf{x}}}(\mathbf{x})
    +L_{\mathbf{G}}V_{\hat{\mathbf{x}}}(\mathbf{x})\mathbf{u} 
    \leq -\gamma(V_{\hat{\mathbf{x}}}(\mathbf{x})) 
    + \delta_{\hat{\mathbf{x}}}(\mathbf{x}),\nonumber\\
    &L_{\mathbf{f}}h_q(\mathbf{x})+L_{\mathbf{G}}h_q(\mathbf{x})\mathbf{u} 
    \geq -\omega\alpha(h_q(\mathbf{x})),\nonumber\\
    &\|\mathbf{u}\| \leq u_\mathrm{max},
\end{align}
where $p \in \mathbb{R}_{>0}$. This QP is feasible within the interior of the active safe 
half-space, as the CLF constraint has been relaxed to respect the input bound and the term 
$\omega\alpha(h_q(\mathbf{x}))$ can take any value when $h_q(\mathbf{x}) > 0$. At the boundary of 
the active safe

\newpage

\noindent half-space, the CLF and CBF constraints remain compatible, and the CLF constraint 
still respects the input bound. However, the three constraints may not always be jointly compatible 
at the boundary because $\omega$ becomes ineffective when $h_q(\mathbf{x}) = 0$. Nevertheless, it 
should be noted that this potential infeasibility arises solely due to the input bound, and if 
$u_\mathrm{max}$ is sufficiently large, the QP remains feasible at the boundary. Moreover, as 
$u_\mathrm{max} \rightarrow \infty$, we recover the unbounded controller from 
\eqref{Eq:CLF-CBF-QP-Compatible}.


\section{Backstepping the Hybrid Control Solution} \label{Sec:Backstepping}

This section extends the proposed hybrid control strategy to higher-order strict-feedback systems, 
as defined in Problem 2. This is achieved by extending the flow and jump dynamics to account for 
the additional subsystems, and the main difference lies in the design of the subproblem controller, 
which will now be designed using a joint CLF-CBF backstepping approach.

The hybrid control strategy can now be described using an auxiliary hybrid system with the 
following flow dynamics:
\begin{equation}
    \begin{bmatrix}
        \Dot{\mathbf{x}}\\
        \Dot{\mathbf{z}}_1\\
        \vdots\\
        \Dot{\mathbf{z}}_r\\
        \Dot{\bm{\xi}}
    \end{bmatrix} =
    \begin{bmatrix}
        \mathbf{f}(\mathbf{x}) + \mathbf{G}(\mathbf{x})\mathbf{z}_1\\
        \mathbf{f}_1(\mathbf{x}, \bm{\eta}_1) + \mathbf{G}_1(\mathbf{x}, \bm{\eta}_1)\mathbf{z}_2\\
        \vdots\\
        \mathbf{f}_r(\mathbf{x}, \bm{\eta}_r) + \mathbf{G}_r(\mathbf{x}, \bm{\eta}_r)
        \mathbf{k}_{\bm{\xi}, r}(\mathbf{x}, \bm{\eta}_r)\\
        \mathbf{0}
    \end{bmatrix},\, (\mathbf{x}, \bm{\xi}) \in \mathcal{F},
    \label{Eq:HybridSystemFlow2}
\end{equation}
and the following jump dynamics:
\begin{equation}
    \begin{bmatrix}
        \mathbf{x}^+\\
        \bm{\eta}_r^+\\
        \bm{\xi}^+
    \end{bmatrix} =
    \begin{bmatrix}
        \mathbf{x}\\
        \bm{\eta}_r\\
        \mathbf{s}(\mathbf{x}, \bm{\xi})
    \end{bmatrix},\quad (\mathbf{x}, \bm{\xi}) \in \mathcal{J},
    \label{Eq:HybridSystemJump2}
\end{equation}
where the flow set $\mathcal{F}$, the jump set $\mathcal{J}$, and the update function $\mathbf{s}$ 
are defined as in Section \ref{Sec:HybridController}. Moreover, 
$\mathbf{k}_{\bm{\xi}, r}: \mathbb{R}^{n\times \ell_r} \rightarrow \mathbb{R}^m$ represents a 
locally Lipschitz continuous controller that solves the safe stabilization subproblem defined by
$\bm{\xi} = (\hat{\mathbf{x}}, q) \in \Xi$. More specifically, we now require the controller 
$\mathbf{k}_{\bm{\xi}, r}$ to render the interior of the safe half-space $q$ forward invariant and 
the target point $\hat{\mathbf{x}}$ asymptotically stable with its region of attraction including 
the interior of the safe half-space. Note that we now only require forward invariance of the 
interior of the safe half-space, as it is not possible to design a controller that ensures safety 
for every initial condition of the state $\bm{\eta}_r$ when the initial state $\mathbf{x}$ lies on 
the safe set boundary. With this formulation in place, we can readily extend the result of 
Theorem \ref{Th:MyTheorem1} to higher-order strict-feedback systems.

\begin{theorem} \label{Th:MyTheorem2}
    Consider the control-affine system \eqref{Eq:StrictFeedbackSystem}. If there exists a 
    locally Lipschitz controller 
    $\mathbf{k}_{\bm{\xi}, r}: \mathbb{R}^{n\times \ell_r} \rightarrow \mathbb{R}^m$ which, for the 
    subproblem defined by $\bm{\xi} = (\hat{\mathbf{x}}, q) \in \Xi$, renders the interior of the 
    safe half-space $q$ forward invariant and the target point $\hat{\mathbf{x}}$ asymptotically 
    stable with its region of attraction including the interior of the safe half-space $q$, then 
    the hybrid control strategy described by 
    \eqref{Eq:HybridSystemFlow2}-\eqref{Eq:HybridSystemJump2} renders the interior of the safe set 
    $\mathcal{C}$ in \eqref{Eq:HybridSafeSet} forward invariant and $\hat{\mathbf{x}}$ an 
    asymptotically stable equilibrium point with its region of attraction including 
    $\mathrm{int}(\mathcal{C})$, with respect to the top-level subsystem of 
    \eqref{Eq:StrictFeedbackSystem}.
\end{theorem}

\begin{proof}
    Similar to the proof of Theorem 3, by noting that safety is strictly preserved during jumps, 
    i.e., $h_{q^+}(\mathbf{x}) > 0$.
\end{proof}

To complete the extension of the hybrid approach to higher-order systems, it now remains to detail 
the design of the controller $\mathbf{k}_{\bm{\xi}, r}$, which is presented in the following 
subsections.


\subsection{Backstepping Overview}

This subsection provides a brief overview of CLF and CBF backstepping, laying the groundwork for 
the subproblem controller design. For notation simplicity, we consider a second-order 
strict-feedback system of the form
\begin{align}
    \Dot{\mathbf{x}} &= \mathbf{f}(\mathbf{x}) + \mathbf{G}(\mathbf{x})\mathbf{z}, 
    \label{Eq:BackstepSystem1}\\
    \Dot{\mathbf{z}} &= \mathbf{f}_1(\mathbf{x}, \mathbf{z}) 
    + \mathbf{G}_1(\mathbf{x}, \mathbf{z})\mathbf{u}, \label{Eq:BackstepSystem2}
\end{align}
which serves as a basis for extending the approach to higher-order systems by recursively applying 
the following steps.

\begin{theorem}[CLF Backstepping \cite{taylor2022safe}] \label{Th:CLFBackstepping}
    Let $V: \mathbb{R}^n \rightarrow \mathbb{R}_{\geq 0}$ be a CLF for the subsystem 
    \eqref{Eq:BackstepSystem1}, and let $\mathbf{k}: \mathbb{R}^n \rightarrow \mathbb{R}^p$ be a 
    continuously differentiable controller such that
    \begin{equation}
        L_{\mathbf{f}}V(\mathbf{x}) + L_{\mathbf{G}}V(\mathbf{x})\mathbf{k}(\mathbf{x}) 
        \leq -\gamma(V(\mathbf{x}))
        \label{Eq:CLFBackstepping}
    \end{equation}
    for all $\mathbf{x} \in \mathbb{R}^n$, where 
    $\gamma: \mathbb{R}_{\geq 0} \rightarrow \mathbb{R}_{\geq 0}$ is a class-$\mathcal{K}$ 
    function. Then, the function 
    $V_1: \mathbb{R}^n\times\mathbb{R}^p \rightarrow \mathbb{R}_{\geq 0}$, 
    defined as
    \begin{equation}
        V_1(\mathbf{x}, \mathbf{z}) = V(\mathbf{x}) + 
        \frac{1}{2\beta_v}\|\mathbf{z} - \mathbf{k}(\mathbf{x})\|^2
    \end{equation}
    for all $(\mathbf{x}, \mathbf{z}) \in \mathbb{R}^n\times\mathbb{R}^p$, with 
    $\beta_v \in \mathbb{R}_{>0}$, is a CLF for the overall system 
    \eqref{Eq:BackstepSystem1}-\eqref{Eq:BackstepSystem2} with a corresponding class-$\mathcal{K}$ 
    function $\gamma_1: \mathbb{R}_{\geq 0} \rightarrow \mathbb{R}_{\geq 0}$ such that 
    $\gamma_1(s) < \gamma(s)$ for all $s \in \mathbb{R}_{>0}$.
\end{theorem}

\begin{theorem}[CBF Backstepping \cite{taylor2022safe}] \label{Th:CBFBackstepping}
    Let $h: \mathbb{R}^n \rightarrow \mathbb{R}$ be a CBF for the subsystem 
    \eqref{Eq:BackstepSystem1}, and let $\mathbf{k}: \mathbb{R}^n \rightarrow \mathbb{R}^p$ 
    denote a continuously differentiable controller such that
    \begin{equation}
        L_{\mathbf{f}}h(\mathbf{x}) + L_{\mathbf{G}}h(\mathbf{x})\mathbf{k}(\mathbf{x})
        > -\alpha(h(\mathbf{x}))
        \label{Eq:CBFBackstepping}
    \end{equation}
    for all $\mathbf{x} \in \mathbb{R}^n$, where $\alpha: \mathbb{R} \rightarrow \mathbb{R}$ is an 
    extended class-$\mathcal{K}_\infty$ function. Then, the function 
    $h_1: \mathbb{R}^n\times\mathbb{R}^p \rightarrow \mathbb{R}$, defined as
    \begin{equation}
        h_1(\mathbf{x}, \mathbf{z}) = h(\mathbf{x})
        - \frac{1}{2\beta_h}\|\mathbf{z} - \mathbf{k}(\mathbf{x})\|^2
    \end{equation}
    for all $(\mathbf{x}, \mathbf{z}) \in \mathbb{R}^n\times\mathbb{R}^p$, with 
    $\beta_h \in \mathbb{R}_{>0}$, is a CBF for the overall system 
    \eqref{Eq:BackstepSystem1}-\eqref{Eq:BackstepSystem2} with an associated extended 
    class-$\mathcal{K}_\infty$ function $\alpha_1: \mathbb{R} \rightarrow \mathbb{R}$ such that 
    $\alpha_1(s) \geq \alpha(s)$ for all $s \in \mathbb{R}$.
\end{theorem}

\begin{remark} \label{Rm:2}
    The CBF $h_1$ can be employed to render its 0-superlevel set forward invariant for the 
    full system \eqref{Eq:BackstepSystem1}-\eqref{Eq:BackstepSystem2}. Therefore, to render the 
    0-superlevel set of $h$ forward invariant for the top-level subsystem 
    \eqref{Eq:BackstepSystem1}, the initial condition $(\mathbf{x}_0, \mathbf{z}_0)$ must satisfy 
    $h_1(\mathbf{x}_0, \mathbf{z}_0) \geq 0$. If the initial state $\mathbf{x}_0$ lies within the 
    interior of the 0-superlevel set of $h$, then this requirement can be satisfied by choosing the 
    gain $\beta_h \in \mathbb{R}_{>0}$ such that
    \begin{equation}
        \beta_h \geq \frac{1}{2h(\mathbf{x}_0)}\|\mathbf{z}_0 - \mathbf{k}(\mathbf{x}_0)\|^2.
    \end{equation}
\end{remark}


\subsection{Subproblem Controller Design}

Consider now the safe stabilization subproblem defined by the auxiliary state 
$\bm{\xi} = (\hat{\mathbf{x}}, q) \in \Xi$. As shown in Section \ref{Sec:SubproblemController}, we 
can construct compatible CLF and CBF constraints for the top-level subsystem of 
\eqref{Eq:StrictFeedbackSystem}. Building on this result, we now show that, starting with 
compatible CLF and CBF constraints for the top-level subsystem, it is possible to recursively 
design compatible CLF and CBF constraints for the full system \eqref{Eq:StrictFeedbackSystem} using 
a joint CLF-CBF backstepping approach. For simplicity of notation, we describe the design for a 
second-order system defined as in \eqref{Eq:BackstepSystem1}-\eqref{Eq:BackstepSystem2}, as the 
extension to higher-order systems follows naturally by repeating the same steps recursively.

To this end, let $V_{\hat{\mathbf{x}}}: \mathbb{R}^n \rightarrow \mathbb{R}_{\geq 0}$ and 
$h_q: \mathbb{R}^n \rightarrow \mathbb{R}$ denote a 

\newpage

\noindent CLF and CBF, respectively, for the top-level subsystem \eqref{Eq:BackstepSystem1}, where 
$V_{\hat{\mathbf{x}}}$ may be defined as in \eqref{Eq:SubproblemCLF} and $h_q$ is defined by 
\eqref{Eq:CBFHalfspace}. Moreover, let 
$\mathbf{k}_{\bm{\xi}}: \mathbb{R}^n \rightarrow \mathbb{R}^p$ be a continuously differentiable 
controller that simultaneously satisfies the conditions \eqref{Eq:CLFBackstepping} and 
\eqref{Eq:CBFBackstepping} with respect to $V_{\hat{\mathbf{x}}}$ and $h_q$. Relying on Theorems 
\ref{Th:CLFBackstepping} and \ref{Th:CBFBackstepping}, we can now construct a CLF 
$V_{\hat{\mathbf{x}}, 1}: \mathbb{R}^n\times\mathbb{R}^p \rightarrow \mathbb{R}_{\geq 0}$ and a CBF 
$h_{q, 1}: \mathbb{R}^n\times\mathbb{R}^p \rightarrow \mathbb{R}$ for the overall system 
\eqref{Eq:BackstepSystem1}-\eqref{Eq:BackstepSystem2} as
\begin{equation}
    \begin{aligned}
        V_{\hat{\mathbf{x}}, 1}(\mathbf{x}, \mathbf{z}) &= V_{\hat{\mathbf{x}}}(\mathbf{x}) 
        + \frac{1}{2\beta_v}\|\mathbf{z} - \mathbf{k}_{\bm{\xi}}(\mathbf{x})\|^2,\\
        h_{q, 1}(\mathbf{x}, \mathbf{z}) &= h_q(\mathbf{x}) 
        - \frac{1}{2\beta_h}\|\mathbf{z} - \mathbf{k}_{\bm{\xi}}(\mathbf{x})\|^2,
    \end{aligned}
    \label{Eq:JointBacksteppingCLFCBF}
\end{equation}
for all $(\mathbf{x}, \mathbf{z}) \in \mathbb{R}^n\times\mathbb{R}^p$, with 
$\beta_v, \beta_h \in \mathbb{R}_{>0}$. These definitions yield the following CLF and CBF 
constraints:
\begin{equation}
    \begin{aligned}
        L_{\bar{\mathbf{f}}_1}V_{\hat{\mathbf{x}}, 1}(\mathbf{x}, \mathbf{z}) + 
        L_{\bar{\mathbf{G}}_1}V_{\hat{\mathbf{x}}, 1}(\mathbf{x}, \mathbf{z})\mathbf{u} 
        &\leq -\gamma_1(V_{\hat{\mathbf{x}}, 1}(\mathbf{x}, \mathbf{z})),\\
        L_{\bar{\mathbf{f}}_1}h_{q, 1}(\mathbf{x}, \mathbf{z}) + 
        L_{\bar{\mathbf{G}}_1}h_{q, 1}(\mathbf{x}, \mathbf{z})\mathbf{u} 
        &\geq -\alpha_1(h_{q, 1}(\mathbf{x}, \mathbf{z})),
    \end{aligned}
    \label{Eq:JointBacksteppingConditions}
\end{equation}
where the fields $\bar{\mathbf{f}}_1$ and $\bar{\mathbf{G}}_1$ are defined as
\begin{equation}
    \begin{aligned}
        \bar{\mathbf{f}}_1(\mathbf{x}, \mathbf{z}) &=
        \begin{bmatrix}
            \mathbf{f}(\mathbf{x}) + \mathbf{G}(\mathbf{x})\mathbf{z}\\
            \mathbf{f}_1(\mathbf{x}, \mathbf{z})
        \end{bmatrix},\\
        \bar{\mathbf{G}}_1(\mathbf{x}, \mathbf{z}) &=
        \begin{bmatrix}
            \mathbf{0}\\
            \mathbf{G}_1(\mathbf{x}, \mathbf{z})
        \end{bmatrix},
    \end{aligned}
\end{equation}
for all $(\mathbf{x}, \mathbf{z}) \in \mathbb{R}^n\times\mathbb{R}^p$. Additionally,
$\gamma_1: \mathbb{R}_{\geq 0} \rightarrow \mathbb{R}_{\geq 0}$ is a class-$\mathcal{K}$ function 
associated with the CLF and $\alpha_1: \mathbb{R} \rightarrow \mathbb{R}$ is an extended 
class-$\mathcal{K}_\infty$ function corresponding to the CBF.

Now, by recognizing that the following relation holds:
\begin{equation}
    L_{\bar{\mathbf{G}}_1}h_{q, 1}(\mathbf{x}, \mathbf{z}) = 
    -\frac{\beta_v}{\beta_h}L_{\bar{\mathbf{G}}_1}V_{\hat{\mathbf{x}}, 1}(\mathbf{x}, \mathbf{z}),
    \label{Eq:JointBacksteppingConditions2}
\end{equation}
and by substituting \eqref{Eq:JointBacksteppingConditions2} into 
\eqref{Eq:JointBacksteppingConditions}, we conclude that the CLF and CBF constraints in 
\eqref{Eq:JointBacksteppingConditions} are equivalent to
\begin{equation}
    \scalebox{1}{$
    L_{\bar{\mathbf{G}}_1}V_{\hat{\mathbf{x}}, 1}(\mathbf{x}, \mathbf{z})\mathbf{u} \leq 
    \min\left\{-F_{V_{\hat{\mathbf{x}}, 1}}(\mathbf{x}, \mathbf{z}),
    \frac{\beta_h}{\beta_v}F_{h_{q, 1}}(\mathbf{x}, \mathbf{z})\right\}$}
\end{equation}
for all $(\mathbf{x}, \mathbf{z}) \in \mathbb{R}^n\times\mathbb{R}^p$, where
\begin{equation}
    \begin{aligned}
        F_{V_{\hat{\mathbf{x}}, 1}}(\mathbf{x}, \mathbf{z}) &= 
        L_{\bar{\mathbf{f}}_1}V_{\hat{\mathbf{x}}, 1}(\mathbf{x}, \mathbf{z})
        + \gamma_1(V_{\hat{\mathbf{x}}, 1}(\mathbf{x}, \mathbf{z})),\\
        F_{h_{q, 1}}(\mathbf{x}, \mathbf{z}) &= L_{\bar{\mathbf{f}}_1}
        h_{q, 1}(\mathbf{x}, \mathbf{z})
        + \alpha_1(h_{q, 1}(\mathbf{x}, \mathbf{z})).
    \end{aligned}
\end{equation}
Therefore, as the CLF and CBF constraints in \eqref{Eq:JointBacksteppingConditions} are equivalent 
to a single linear inequality, they are mutually satisfiable for all 
$(\mathbf{x}, \mathbf{z}) \in \mathbb{R}^n\times\mathbb{R}^p$. This means that if it is possible to 
design compatible CLF and CBF constraints for the top-level subsystem of 
\eqref{Eq:StrictFeedbackSystem}, we can recursively apply the previous joint CLF-CBF backstepping 
approach to establish compatible CLF and CBF constraints for the overall system 
\eqref{Eq:StrictFeedbackSystem}.

However, it is important to highlight that the intermediate controller $\mathbf{k}_{\bm{\xi}}$, 
used for constructing the CLF and CBF in \eqref{Eq:JointBacksteppingCLFCBF}, must be continuously 
differentiable. This means that we cannot design $\mathbf{k}_{\bm{\xi}}$ by means of a QP as in 
Section \ref{Sec:SubproblemController} because QPs typically yield only locally Lipschitz 
continuous controllers. 

To meet the smoothness requirement, we adopt the approach introduced in \cite{ong2019universal}, 
which constructs a smooth controller based on Gaussian-weighted centroids. To this end, suppose we 
have established compatible CLF and CBF constraints for the top-level subsystem 
\eqref{Eq:BackstepSystem1}, which define the sets of controls
\begin{equation}
    \scalebox{0.905}{$
    \begin{aligned}
        K_{V_{\hat{\mathbf{x}}}}(\mathbf{x}) &= \{\mathbf{z} \in \mathbb{R}^p:
        L_{\mathbf{f}}V_{\hat{\mathbf{x}}}(\mathbf{x}) 
        + L_{\mathbf{G}}V_{\hat{\mathbf{x}}}(\mathbf{x})\mathbf{z} 
        \leq -\gamma(V_{\hat{\mathbf{x}}}(\mathbf{x}))\},\\
        K_{h_q}(\mathbf{x}) &= \{\mathbf{z} \in \mathbb{R}^p:
        L_{\mathbf{f}}h_q(\mathbf{x}) + L_{\mathbf{G}}h_q(\mathbf{x})\mathbf{z} 
        \geq -\alpha(h_q(\mathbf{x}))\},
    \end{aligned}$}
    \label{Eq:JointBacksteppingConditions3}
\end{equation}
for all $\mathbf{x} \in \mathbb{R}^n$. Following the design approach from \cite{ong2019universal}, 
we

\newpage

\noindent can formulate a top-level controller $\mathbf{k}_{\bm{\xi}}$ as
\begin{equation}
    \begin{aligned}
        \mathbf{k}_{\bm{\xi}}(\mathbf{x}) &= \zeta(\rho(\mathbf{x}))
        \left(\bm{\mu}\big(K_{V_{\hat{\mathbf{x}}}}(\mathbf{x})\big) 
        + \bm{\mu}\left(K_{h_q}(\mathbf{x})\right)\right)\\
        &+ (1-\zeta(\rho(\mathbf{x})))
        \bm{\mu}\left(K_{V_{\hat{\mathbf{x}}}}(\mathbf{x}) \cap K_{h_q}(\mathbf{x})\right)
    \end{aligned}
    \label{Eq:GaussianController}
\end{equation}
for all $\mathbf{x} \in \mathbb{R}^n$, where $\zeta: \mathbb{R} \rightarrow [0, 1]$ is a smooth 
partition of the unit step function, defined as
\begin{equation}
    \zeta(s) = 
    \begin{cases}
        0, &\text{if } s \leq 0,\\
        \left(1+\dfrac{\exp(1/s)}{\exp(1/(s-1))}\right)^{-1}, &\text{if } 0 < s < 1,\\
        1, &\text{if } s \geq 1,
    \end{cases}
\end{equation}
and the function $\rho: \mathbb{R}^n \rightarrow [-1, 1]$, defined as
\begin{equation}
    \rho(\mathbf{x}) = 
    \frac{L_{\mathbf{G}}V_{\hat{\mathbf{x}}}(\mathbf{x})L_{\mathbf{G}}h_q(\mathbf{x})^\top}{
    \|L_{\mathbf{G}}V_{\hat{\mathbf{x}}}(\mathbf{x})\|\|L_{\mathbf{G}}h_q(\mathbf{x})\|},
\end{equation}
encodes the angle between $L_{\mathbf{G}}V_{\hat{\mathbf{x}}}(\mathbf{x})$ and 
$L_{\mathbf{G}}h_q(\mathbf{x})$. Moreover, 
$\bm{\mu}: \mathcal{P}(\mathbb{R}^n) \rightarrow \mathbb{R}^n$ represents the Gaussian-weighted 
centroid function, defined as 
\begin{equation}
    \bm{\mu}(\mathcal{S}) = \frac{\int_\mathcal{S}\mathbf{z}
    \exp(-\|\mathbf{z}\|^2/(2\varsigma))d\mathbf{z}}
    {\int_\mathcal{S}\exp(-\|\mathbf{z}\|^2/(2\varsigma))d\mathbf{z}}
    \label{Eq:GaussianCentroid}
\end{equation}
for every $\mathcal{S} \in \mathcal{P}(\mathbb{R}^n)$, with $\varsigma \in \mathbb{R}_{>0}$, 
which can be expressed in closed form when $\mathcal{S}$ is a half-space \cite{tallis1961moment}, 
\cite{tallis1965plane}. The controller defined by \eqref{Eq:GaussianController} respects both the 
CLF and CBF constraints in \eqref{Eq:JointBacksteppingConditions3}, and it is smooth over 
$\mathbb{R}^n \setminus \{\hat{\mathbf{x}}\}$, provided that the system dynamics, the gradients of 
the CLF and CBF, and the functions $\gamma$ and $\alpha$ are all smooth. Although the controller 
may be nonsmooth at $\hat{\mathbf{x}}$, this poses no issues when 
$\hat{\mathbf{x}} \neq \bar{\mathbf{x}}$ because a jump will occur before $\hat{\mathbf{x}}$ is 
reached. The only problematic scenario arises when $\hat{\mathbf{x}} = \bar{\mathbf{x}}$ and no 
more jumps will occur, as the top-level state $\mathbf{x}$ may overshoot through 
$\bar{\mathbf{x}}$. Nevertheless, this issue can be remedied by slightly relaxing the stabilization 
objective within an arbitrarily small region around $\bar{\mathbf{x}}$ \cite{taylor2022safe}.

\begin{remark}[Maintaining Safety across Jumps]
    To maintain safety when an update occurs, the overall state $(\mathbf{x}, \mathbf{z})$ at the 
    moment of the jump must satisfy $h_{q^+, 1}(\mathbf{x}, \mathbf{z}) \geq 0$. As shown in 
    Section \ref{Sec:HybridController}, the top-level state $\mathbf{x}$ will be within the 
    interior of the next active safe half-space, so that $h_{q^+}(\mathbf{x}) > 0$ at the moment of 
    the jump. Hence, as mentioned in Remark \ref{Rm:2}, we can select the gain $\beta_h$ such 
    that
    \begin{equation}
        \beta_h^+ \geq \frac{1}{2h_{q^+}(\mathbf{x})}
        \left\|\mathbf{z} - \mathbf{k}_{\bm{\xi}^+}(\mathbf{x})\right\|^2
    \end{equation}
    to preserve safety across transitions. This means that in practice, the value of $\beta_h$ may 
    have to change when a jump occurs. However, for simplicity of notation, this parameter 
    dependence was omitted in the formulation from 
    \eqref{Eq:HybridSystemFlow2}–\eqref{Eq:HybridSystemJump2}. If we have a desired nominal value 
    $\bar{\beta}_h$, we can then update $\beta_h$ as
    \begin{equation}
        \beta_h^+ = \max\left\{\bar{\beta}_h, \frac{1}{2h_{q^+}(\mathbf{x})}
        \left\|\mathbf{z} - \mathbf{k}_{\bm{\xi}^+}(\mathbf{x})\right\|^2\right\}.
    \end{equation}
    For higher-order systems defined as in \eqref{Eq:StrictFeedbackSystem}, there will be $r$ 
    parameters $\beta_h$, one for each level of the backstepping design, and the process has to be 
    applied recursively to ensure that $h_{q^+, r}(\mathbf{x}, \bm{\eta_r}) \geq 0$ when a jump 
    occurs.
\end{remark}

Since compatible CLF and CBF can be constructed for the top‐level subsystem of 
\eqref{Eq:StrictFeedbackSystem}, we can recursively apply the joint CLF-CBF backtepping approach to 
establish compatible CLF and CBF constraints for the overall system 
\eqref{Eq:StrictFeedbackSystem}. Once these

\newpage

\noindent constraints have been obtained, we can design an optimization‐based controller 
$\mathbf{k}_{\bm{\xi}, r}$ as follows:
\begin{equation} 
    \scalebox{0.915}{$
    \begin{aligned}
        &\mathbf{k}_{\bm{\xi}, r}(\mathbf{x}, \bm{\eta}_r) = 
        \,\, \underset{\mathbf{u} \in \mathbb{R}^m}{\arg\min}\,\, 
        \frac{1}{2}\|\mathbf{u}\|^2 \\
        \text{s.t.}\,\, 
        &L_{\bar{\mathbf{f}}_r}V_{\hat{\mathbf{x}}, r}(\mathbf{x}, \bm{\eta}_r) + 
        L_{\bar{\mathbf{G}}_r}V_{\hat{\mathbf{x}}, r}(\mathbf{x}, \bm{\eta}_r)\mathbf{u} 
        \leq -\gamma_r(V_{\hat{\mathbf{x}}, r}(\mathbf{x}, \bm{\eta}_r)),\\
        &L_{\bar{\mathbf{f}}_r}h_{q, r}(\mathbf{x}, \bm{\eta}_r) + 
        L_{\bar{\mathbf{G}}_r}h_{q, r}(\mathbf{x}, \bm{\eta}_r)\mathbf{u} 
        \geq -\alpha_r(h_{q, r}(\mathbf{x}, \bm{\eta}_r)),
    \end{aligned}$}
    \label{Eq:SubproblemController2}
\end{equation}
where $\gamma_r$ is a class-$\mathcal{K}$ function associated with the CLF and $\alpha_r$ is an 
extended class-$\mathcal{K}_\infty$ function associated with the CBF. This controller is guaranteed 
to be locally Lipschitz continuous over the state space, except at the zero of the CLF 
\cite{morris2013sufficient}, \cite{jankovic2018robust}. If input constraints have to considered, 
the optimal-decay QP formulation discussed in Section \ref{Sec:SubproblemController} has to be 
employed.


\section{Simulation Results} \label{Sec:Results}

This section presents simulation results illustrating the trajectories achieved with the proposed 
hybrid control strategy. We also discuss the advantages of our approach compared to the one 
presented in \cite{marley2024hybrid}, which, to the best of our knowledge, is the most similar 
alternative available in the literature.


\subsection{First-Order Dynamics}

We begin by considering a system with first-order dynamics. Particularly, for simplicity, we 
consider the integrator system defined by \eqref{Eq:SingleIntegrator}, and we apply the hybrid 
control law detailed in Section \ref{Sec:HybridController} to stabilize the system to a desired 
equilibrium point while avoiding a convex polytope. Fig. \ref{Fig:Results1} displays several 
examples of the trajectories and temporal profiles obtained for different polytopes in a 
$2$-dimensional setting ($n = 2$).

\begin{figure}[t]
    \centering 
    \subfloat[Pentagon, fixed initial state, variation of the initial active half-space]{
        \includegraphics[width=0.97\linewidth]{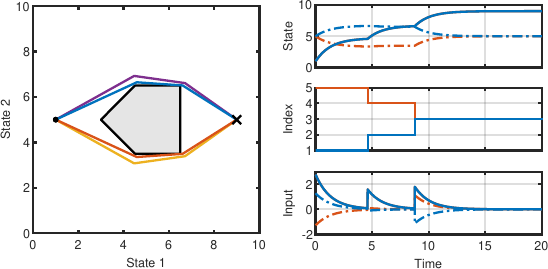}}\\
    \subfloat[Square, fixed initial state, variation of the initial reference direction]{
        \includegraphics[width=0.97\linewidth]{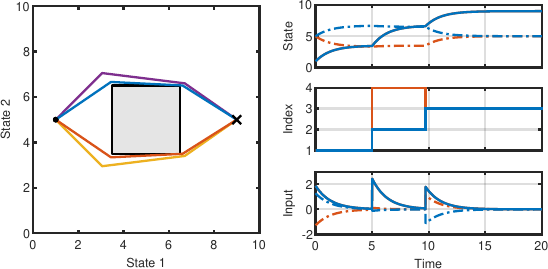}}\\
    \subfloat[Triangle, variation of the initial state]{
        \includegraphics[width=0.97\linewidth]{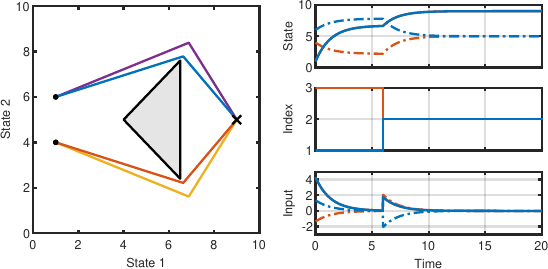}}
    \caption{Examples of system trajectories and the corresponding temporal profiles under the 
    hybrid control law from Section \ref{Sec:HybridController} for three different polytopes while 
    considering a fixed desired equilibrium point. The plots on the left display trajectories 
    obtained with $\mu = 0.2$ (blue and orange) and $\mu = 1$ (purple and yellow) for a fixed 
    $\sigma = 0.1$. The plots on the right display the respective time evolution of the state, the 
    input, and the index of the active half-space for the blue and orange trajectories. The initial 
    state is denoted as \xtarget and the desired equilibrium point as \xchaserinit.}
    \label{Fig:Results1}
\end{figure}

Fig. \ref{Fig:Results1} (a) presents an example where the initial state allows for two options for 
the initial active safe half-space. Therefore, depending on the selected initialization, two 
different trajectories can be achieved. In contrast, Fig. \ref{Fig:Results1} (b) depicts a scenario 
where the initial active safe half-space is well defined, but two possible directions can be 
associated with it. More specifically, Fig. \ref{Fig:Results1} (b) presents an example in which the 
reference direction is collinear with the normal vector of the initial active half-space, and thus, 
decisiveness is achieved as in \eqref{Eq:DecisivenessDirection} by selecting either 
$\bm{\epsilon} = (0, 1)$ or $\bm{\epsilon} = (0, -1)$. Hence, similar to Fig. 
\ref{Fig:Results1} (a), different trajectories can be produced depending on the chosen direction. 
In addition, Fig. \ref{Fig:Results1} (c) shows the results obtained for a triangular polytope 
across distinct initial states, where the initial half-space and associated direction are clearly 
defined. 

Fig. \ref{Fig:Results1} also demonstrates the impact of the synergy gap $\mu$ on the resulting 
trajectories, as all the examples are presented for two different values of $\mu$. As it can be 
noticed, as the synergy gap increases, the trajectories become more conservative since the 
intermediate target points are placed farther away from the polytope. Meanwhile, the hysteresis 
width $\sigma$ determines how deeply the system trajectories must go into the next safe half-space 
before an update (jump) occurs; however, this parameter remains constant throughout the 
simulations. 

In all the cases displayed in Fig. \ref{Fig:Results1}, the system successfully avoids the polytopic 
region and reaches the desired equilibrium point, as guaranteed by Theorem \ref{Th:MyTheorem1}. In 
particular, we highlight that, for the cases illustrated in Figs.  \ref{Fig:Results1} (a) and 
\ref{Fig:Results1} (b), a deadlock situation would occur if a continuous control approach would be 
considered, such as the one discussed in Section \ref{Sec:IllustrativeExamplesSmoothMax}.


\subsection{Comparison with Previous Work}

The most similar alternative available in the literature is the one recently proposed in
\cite{marley2024hybrid}, which also consists of a hybrid feedback approach relying on a polytopic 
avoidance domain. However, the alternative from \cite{marley2024hybrid} can be characterized as a 
hybrid CBF-only method since only the active safe half-space is updated when a jump occurs, and the 
target point remains fixed at the desired final equilibrium point. As a result, for each active 
half-space that does not contain the desired equilibrium point, there exists an induced deadlock 
point on its boundary to which the trajectory converges. Consequently, a significant limitation of 
the approach proposed in \cite{marley2024hybrid} is that deadlock resolution is only achievable for 
certain polytopes where all the induced equilibria are in positions that allow for switching the 
active safe half-space. This means that a specific polytope must be carefully designed to enclose 
the actual unsafe region while also satisfying this condition. However, such a design will also 
only be valid for a particular set of desired equilibrium points, as the positions of the induced 
equilibria depend on the desired equilibrium point through the CLF.

In contrast, our method takes advantage of the fact that, for a given safe half-space, it is 
possible to design a CLF-CBF controller based on compatible CLF and CBF constraints, which ensures 
convergence to any desired target point that belongs to the half-space. Hence, rather than 
directing trajectories toward fixed induced equilibrium points, our approach automatically assigns 
a target point to each active safe half-space in such a way that it produces a sequence of 
setpoints that converge to the desired equilibrium point. This characterizes our method as a hybrid 
CLF-CBF approach since both the active setpoint and safe half-space are updated when a jump occurs. 
Consequently, the primary advantage of our strategy is that global asymptotic stabilization and 
safety are ensured for any convex polytope. Thus, for a given unsafe set, the only remaining task 
is fitting any convex polytope to that region. Furthermore, the proposed approach offers greater 
flexibility and configurability, enabling adjustments to the conservativeness of the trajectories 
and the strategies employed to ensure decisiveness.

The previously mentioned advantages are highlighted in Fig. \ref{Fig:Results2}, which compares the 
trajectories obtained with our approach and the one from \cite{marley2024hybrid} across different 
polytopes and desired equilibrium points. In the first example, shown in Fig. \ref{Fig:Results2}
(a), the system successfully avoids the polytopic region and reaches the desired equilibrium 
point under both strategies for every initial state. This happens because, as it can be noticed, 
all the induced equilibrium points to which the trajectories under the method from 
\cite{marley2024hybrid} may converge lie within more than one safe half-space. However, as 
displayed in Fig. \ref{Fig:Results2} (b), for a different desired equilibrium point, not all the 
induced equilibria satisfy that condition, leading to instances where deadlock resolution is not 
achieved and the objective is not completed. In addition, Fig. \ref{Fig:Results2} (c) presents a 
scenario where, for a simple square, deadlock resolution is also not achieved for every initial 
state under the approach from \cite{marley2024hybrid}. In contrast, using the strategy proposed in 
this paper, the system successfully avoids the polytope and reaches the desired equilibrium point 
in all the instances shown in Fig. \ref{Fig:Results2}, as guaranteed by Theorem 
\ref{Th:MyTheorem1}. Also, as can be noticed, the trajectories produced by our approach directly 
converge to the auxiliary setpoints, rather than initially converging toward the boundary of the 
active half-space and subsequently performing an unnecessary curve.

\begin{figure}[t]
    \centering 
    \subfloat[Proposed (left) vs. \cite{marley2024hybrid} (right), pentagon, both successful]{
        \includegraphics[width=0.98\linewidth]{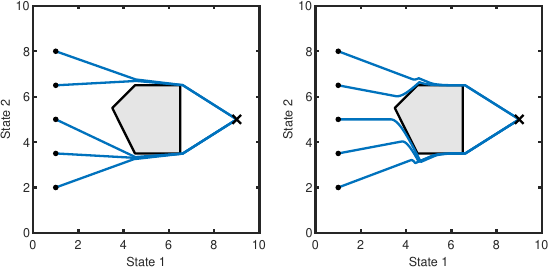}}\\
    \subfloat[Proposed (left) vs. \cite{marley2024hybrid} (right), pentagon, 
    \cite{marley2024hybrid} fails]{
        \includegraphics[width=0.98\linewidth]{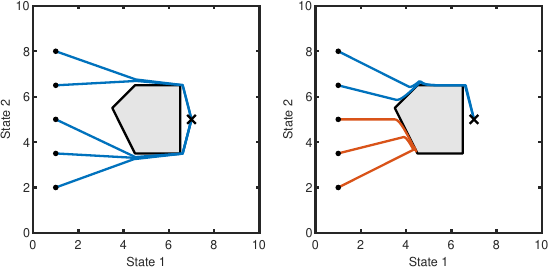}}\\
    \subfloat[Proposed (left) vs. \cite{marley2024hybrid} (right), square, 
    \cite{marley2024hybrid} fails]{
        \includegraphics[width=0.98\linewidth]{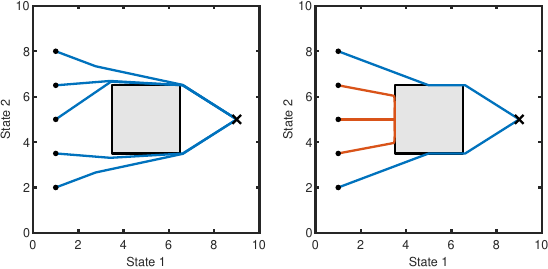}}
    \caption{Comparison between the system trajectories generated using the hybrid feedback 
    strategy detailed in Section \ref{Sec:HybridController} (left) and the ones obtained with the 
    approach proposed in \cite{marley2024hybrid} (right) across different polytopes and desired 
    equilibrium points. Blue trajectories correspond to cases where the system successfully avoids 
    the polytope and reaches the desired equilibrium point. Meanwhile, orange trajectories indicate 
    cases in which the system incurs in a deadlock situation. The initial state is denoted as 
    \xtarget and the desired equilibrium point as \xchaserinit.}
    \label{Fig:Results2}
\end{figure}


\subsection{Second-Order Dynamics}

Additionally, we now consider the double-integrator system
\begin{equation}
    \begin{aligned}
        \Dot{\mathbf{x}} &= \mathbf{z},\\
        \Dot{\mathbf{z}} &= \mathbf{u},
    \end{aligned}
\end{equation}
where $\mathbf{x}, \mathbf{z}, \mathbf{u} \in \mathbb{R}^n$, and we backstep the hybrid 
control law, as detailed in Section \ref{Sec:Backstepping}, to stabilize the top-level subsystem 
to a desired equilibrium point while avoiding a convex polytope. Fig. \ref{Fig:Results3} revisits 
the examples from Fig. \ref{Fig:Results1}, now for the double-integrator system, and displays the 
top-level system trajectories and the respective temporal profiles obtained for the different 
polytopes in a $2$-dimensional setting ($n = 2$). In all the cases displayed in Fig. 
\ref{Fig:Results3}, the top-level subsystem successfully avoids the polytopic unsafe set and 
reaches the desired equilibrium point. Moreover, as can be noticed, the trajectories presented in 
Fig. \ref{Fig:Results3} are similar to those from Fig. \ref{Fig:Results1}, however, extending to a 
system with second-order dynamics results in smoother trajectories and removes the sharp corners 
seen in Fig. \ref{Fig:Results1}.

\begin{figure}[t]
    \centering 
    \subfloat[Pentagon, fixed initial state, variation of the initial active half-space]{
        \includegraphics[width=0.97\linewidth]{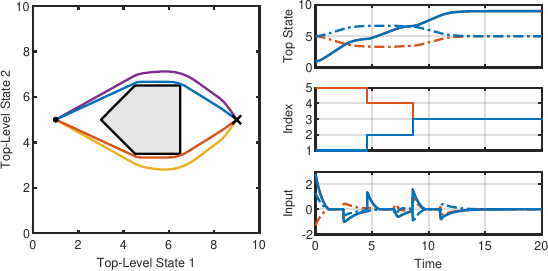}}\\
    \subfloat[Square, fixed initial state, variation of the initial reference direction]{
        \includegraphics[width=0.97\linewidth]{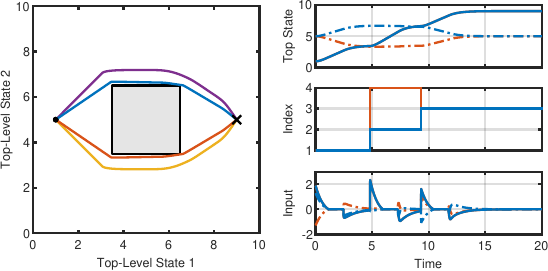}}\\
    \subfloat[Triangle, variation of the initial state]{
        \includegraphics[width=0.97\linewidth]{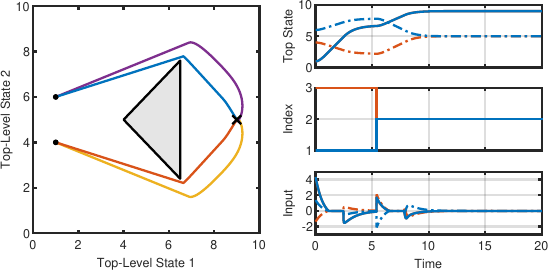}}
    \caption{Examples of top-level system trajectories and the corresponding temporal profiles for 
    the double-integrator system under the hybrid control law from Section \ref{Sec:Backstepping} 
    for three different polytopes while considering a fixed desired equilibrium point. The plots on 
    the left display trajectories obtained with $\mu = 0.2$ (blue and orange) and $\mu = 1$ 
    (purple and yellow) for a fixed $\sigma = 0.1$. The plots on the right display the respective 
    time evolution of the top-level state, the input, and the index of the active safe half-space 
    for the blue and orange trajectories. The initial top-level state is denoted as \xtarget and 
    the desired equilibrium point as \xchaserinit. In all the examples, the system starts at rest.}
    \label{Fig:Results3}
\end{figure}

\begin{figure}[t]
    \centering 
    \subfloat[Square, fixed initial state, variation of the initial reference direction]{
        \includegraphics[width=0.97\linewidth]{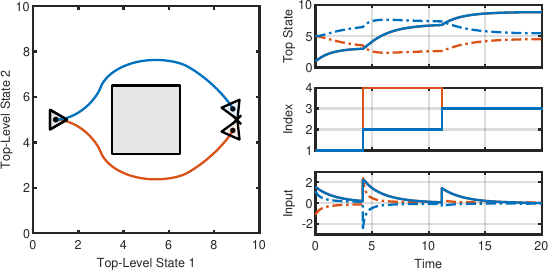}}\\
    \subfloat[Triangle, variation of the initial state]{
        \includegraphics[width=0.97\linewidth]{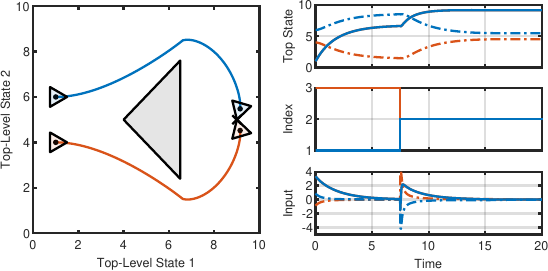}}
    \caption{Examples of top-level system trajectories and the corresponding temporal 
    profiles for the planar unicycle model under the hybrid control law from Section 
    \ref{Sec:HybridController} with the coordinate transformation defined in \eqref{Eq:UnicycleCT}. 
    The left-hand plots show system trajectories obtained with $\mu = 0.2$ and 
    $\sigma = 0.1$. The plots on the right display the respective time evolution of the original 
    position $\mathbf{x}$, the input, and the index of the active safe half-space. The original 
    control point $\mathbf{x}$ corresponds to the center of the triangles depicted in the figure, 
    denoted as \xtarget, and the new control point $\mathbf{x}'$ corresponds to the front vertex of 
    the triangles. The desired equilibrium position is indicated by \xchaserinit.}
    \label{Fig:Results4}
\end{figure}


\subsection{Unicycle Model}

Ultimately, we demonstrate how the proposed hybrid control strategy can be applied to the planar 
unicycle model, which is a nonlinear model that does not directly fall into the classes of systems
considered in this paper.

The planar unicycle model is defined by
\begin{equation}
    \begin{aligned}
        \Dot{\mathbf{x}} &= v(\cos(\theta), \sin(\theta)),\\
        \Dot{\theta} &= \omega,
    \end{aligned}
\end{equation}
with position $\mathbf{x} \in \mathbb{R}^2$, heading angle $\theta \in \mathbb{R}$, and control 
input $\mathbf{u} = (v, \omega) \in \mathbb{R}^2$ consisting of the linear velocity $v$ and angular 
velocity $\omega$. As can be noticed, the unicycle model exhibits a mixed relative degree with 
respect to $\mathbf{x}$ since the control input enters the system dynamics at different levels, and 
it does not directly conform to the system classes considered throughout the paper. Nevertheless, 
we can overcome this by following the strategy adopted in \cite{glotfelter2019hybrid}, which 
considers the control of a point $\mathbf{x}'$ located ahead of $\mathbf{x}$ along the heading 
direction. Specifically, this corresponds to defining the new control point $\mathbf{x}'$ as
\begin{equation}
    \mathbf{x}' = \mathbf{x} + l(\cos(\theta), \sin(\theta)),
    \label{Eq:UnicycleCT}
\end{equation}
where $l \in \mathbb{R}_{>0}$ determines the offset distance. As a result, the system dynamics 
relative to $\mathbf{x}'$ are given by
\begin{equation}
    \Dot{\mathbf{x}}' = \mathbf{R}(\theta)\mathbf{L}\mathbf{u},
    \label{Eq:Unicycle}
\end{equation}
where the function $\mathbf{R}: \mathbb{R} \rightarrow \mathrm{SO}(2)$ is defined as
\begin{equation}
    \mathbf{R}(\theta) =
    \begin{bmatrix}
        \cos(\theta) & -\sin(\theta)\\
        \sin(\theta) & \cos(\theta)
    \end{bmatrix}
\end{equation}
for all $\theta \in \mathbb{R}$, and the matrix $\mathbf{L} \in \mathbb{R}^{2\times 2}$ is defined 
as
\begin{equation}
    \mathbf{L} =
    \begin{bmatrix}
        1 & 0\\
        0 & l
    \end{bmatrix}.
\end{equation}
Thus, with this transformation, we have ended up with a first-order parameter-varying system in 
\eqref{Eq:Unicycle}, where $\mathbf{R}(\theta)\mathbf{L}$ has full row rank for all 
$\theta \in \mathbb{R}$. Therefore, the hybrid control strategy from Section 
\ref{Sec:HybridController} can be applied to the system dynamics \eqref{Eq:Unicycle}. To ensure 
safety with respect to the original position, the active CBF $h_q$ must consider the offset 
distance and be defined as
\begin{equation}
    h_q(\mathbf{x}') = \mathbf{n}_q^\top\mathbf{x}' - d_q - l.
\end{equation}

Fig. \ref{Fig:Results4} revisits two earlier examples, this time applied to the unicycle model 
using the change of coordinates from \eqref{Eq:UnicycleCT}. Due to the nature of the unicycle 
model, the resulting trajectories

\newpage

\noindent are more curvilinear compared to previous systems. Moreover, the original point 
$\mathbf{x}$ (the triangle's center) does not converge exactly to the desired equilibrium point. 
Instead, it converges to a ball of radius equal to the offset distance. This occurs because the new 
control point $\mathbf{x}'$ (the front vertex of the triangle) is the one converging exactly to the 
target. Nevertheless, the conservativeness of this approach can be reduced by selecting a smaller 
offset distance.


\section{Conclusion} \label{Sec:Conclusion}

This paper introduces a hybrid CLF-CBF control framework with global asymptotic stabilization 
properties, overcoming the limitations concerning deadlocks found in the standard CLF-CBF-based 
framework. The proposed solution provides a more flexible and systematic design approach than 
current alternatives available in the literature, ensuring global asymptotic stabilization and 
safety across any bounded convex polytopic avoidance domain. The approach is further 
extended to higher-order systems via a joint CLF-CBF backstepping procedure.

Avenues for further research include extending this method to handle unsafe regions composed of 
multiple polytopes and time-varying unsafe sets. Moreover, an experimental validation with a 
vehicle could be a valuable next step to gather real-world data supporting the method's 
effectiveness.


\bibliographystyle{ieeetr}
\bibliography{Refs}

\begin{IEEEbiography}[
{\includegraphics[width=1in,height=1.25in,keepaspectratio]{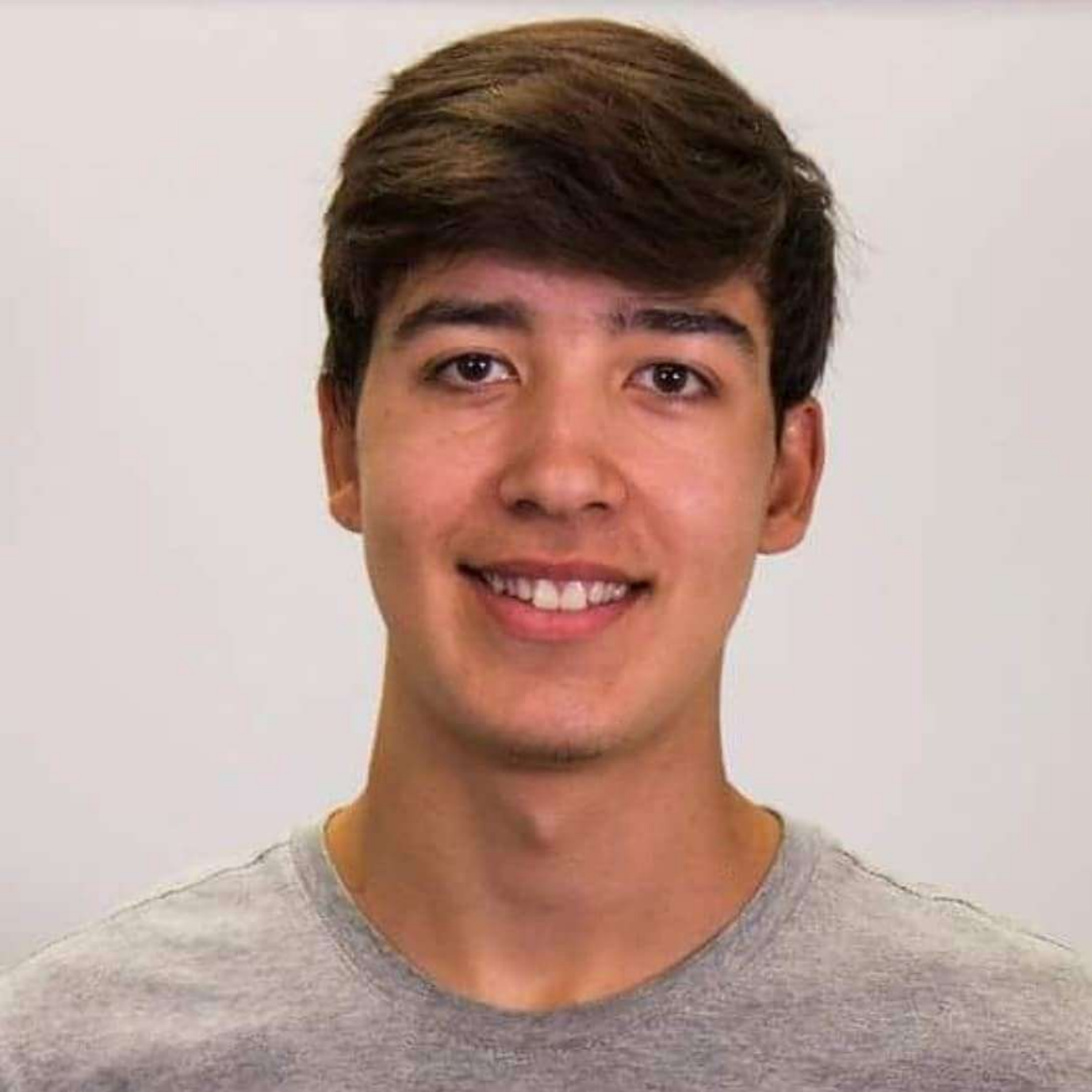}}]
{Hugo Matias} received his B.Sc. and M.Sc. degrees in Electrical and Computer Engineering from the 
Instituto Superior Técnico (IST), Lisbon, Portugal, in 2021 and 2023, respectively. His main area 
of specialization is Control, Robotics and Artificial Intelligence, and his second area of 
specialization is Networks and Communication Systems. He is pursuing a Ph.D. in Electrical and 
Computer Engineering, with a specialization in Systems, Decision and Control, at the School of 
Science and Technology from the NOVA University of Lisbon, Costa da Caparica, Portugal, and he is 
also conducting his research at the Institute for Systems and Robotics (ISR), LARSyS, Lisbon, 
Portugal. His research interests span the fields of nonlinear control and optimization, filtering 
and estimation, and distributed systems.
\end{IEEEbiography}

\vspace{-1mm}

\begin{IEEEbiography}[
{\includegraphics[width=1in,height=1.25in,keepaspectratio]{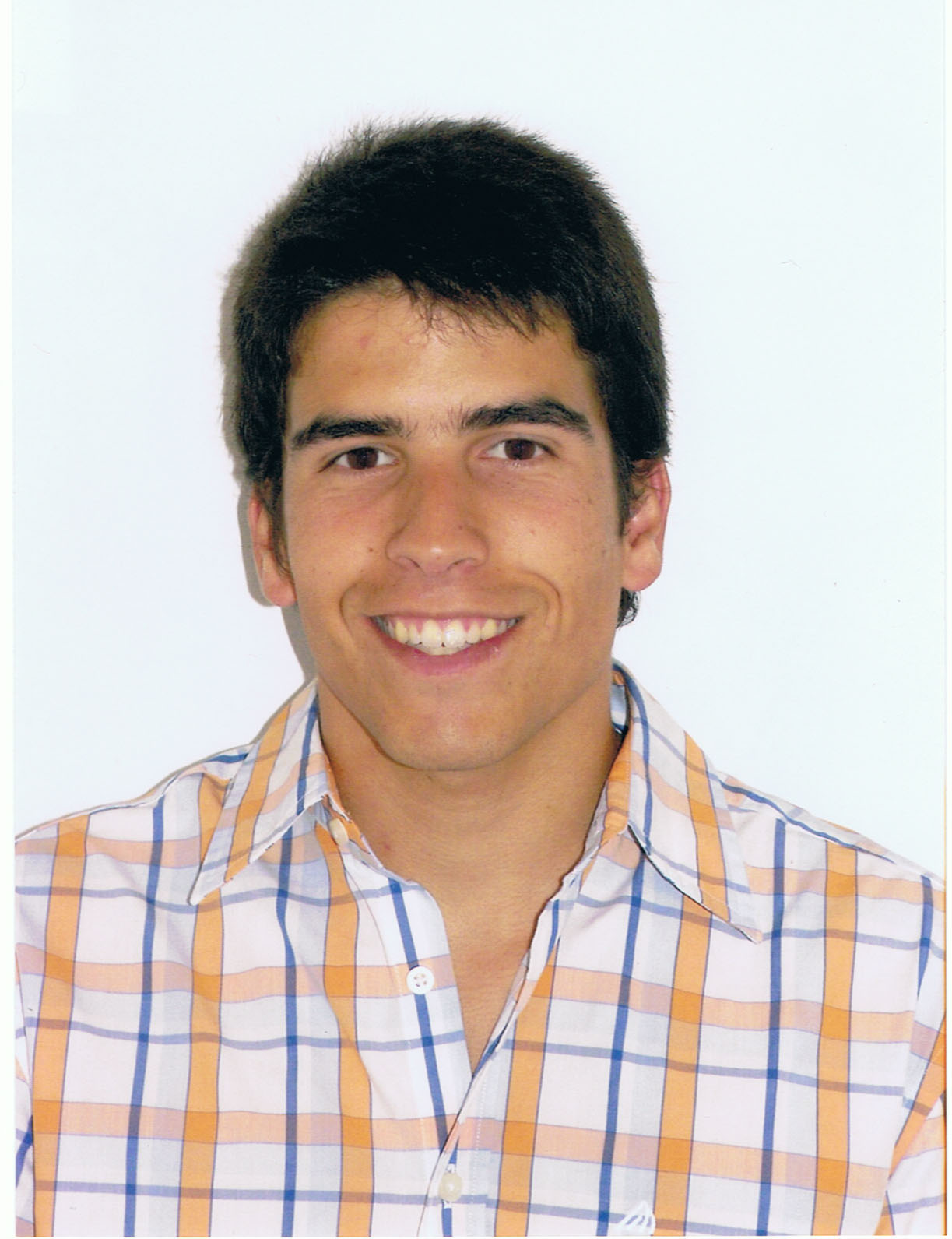}}]
{Daniel Silvestre} received his B.Sc. in Computer Networks in 2008 from the Instituto 
Superior Técnico (IST), Lisbon, Portugal, and his M.Sc. in Advanced Computing in 2009 from the 
Imperial College London, London, United Kingdom. In 2017, he received his Ph.D. (with the highest 
honors) in Electrical and Computer Engineering from the former university. Currently, he is with 
the School of Science and Technology from the Nova University of Lisbon and also with the Institute 
for Systems and Robotics at the Instituto Superior Técnico in Lisbon (PT). His research interests 
span the fields of fault detection and isolation, distributed systems, guaranteed state estimation, 
computer networks, optimal control and nonlinear optimization.
\end{IEEEbiography}

\end{document}